%% file: main.tex
\documentclass[11pt,letterpaper]{article}
\usepackage{times}
\usepackage{a4wide}
\usepackage{amsmath}
\usepackage{bbm}
\usepackage{amsfonts}
\usepackage{amssymb}
\usepackage{amsthm}
\usepackage{verbatim}
\usepackage{enumitem}

\newcounter{all}

\newtheorem{theorem}[all]{Theorem}
\newtheorem*{theorem*}{Theorem}
\newtheorem{proposition}[all]{Proposition}
\newtheorem{lemma}[all]{Lemma}

\newtheorem*{remark*}{Remark}

\def\R{\mathbb{R}}
\def\Rnonneg{\R_p}
\def\Q{\mathbb{Q}}
\def\Qnonneg{\Q_{\geq0}}
\def\sigf#1{:\{0,1\}^{#1}\to\Rnonneg}
\def\vecq{{\boldsymbol q}}
\def\vecx{{\boldsymbol x}}
\def\vecy{{\boldsymbol y}}
\def\vecz{{\boldsymbol z}}
\def\veczero{{\boldsymbol 0}}
\def\vecone{{\boldsymbol 1}}
\def\vecp{{\boldsymbol p}}
\def\pin#1#2{#1_{#2}}
\def\setdiff{\triangle}

\def\nCSP{\operatorname{\#CSP}}

\def\PMp{\mathsf{\#PM}}
\def\PMP{\mathsf{\#PM}}

\def\BIS{\mathsf{\#BIS}}
\def\SAT{\mathsf{\#SAT}}
\def\HC{\mathsf{\#HC}}
\def\Xprob{\mathsf{\#X}}
\def\Yprob{\mathsf{\#Y}}
\def\PIN{\mathrm{PIN}}
\def\NAND{\mathrm{NAND}}
\def\OR{\mathrm{OR}}
\def\EQ{\mathrm{EQ}}
\def\NEQ{\mathrm{NEQ}}
\def\IMP{\mathrm{IMP}}

\def\PM{\mathrm{PM}}

\def\calC{\mathcal{C}}
\def\calF{\mathcal{F}}
\def\calG{\mathcal{G}}

\def\APred{\leq_{AP}}
\def\APeq{=_{AP}}
\def\supp{\operatorname{supp}}
\def\dom{\operatorname{dom}}
\def\deg{\operatorname{deg}}
\def\wt#1#2{\operatorname{wt}^{#1}_{#2}}

\def\Neqconj{\text{NEQ-conj}}
\def\Imconj{\text{IM-conj}}
\def\nCSPs{\nCSP^{\geq 0}}
\def\cwt{\nCSPs_{\leq 2}}
\def\gpin{\Gamma_{\mathrm{pin}}}
\def\scope{\operatorname{scope}}
\def\zp#1{Z_{#1}}
\def\zwp#1#2{Z^{#1}_{#2}}

\makeatletter
\newtheorem*{rep@theorem}{\rep@title}
\newcommand{\newreptheorem}[2]{%
\newenvironment{rep#1}[1]{%
 \def\rep@title{#2 \ref{##1}}%
 \begin{rep@theorem}}%
 {\end{rep@theorem}}}
\makeatother

\newreptheorem{theorem}{Theorem}
\newreptheorem{lemma}{Lemma}

\begin{document}
\title{Degree two approximate Boolean \#CSPs with variable weights}
\author{Colin McQuillan
\thanks{Supported by an EPSRC doctoral training grant.}
\\University of Liverpool
\\ cmcq@liv.ac.uk}

\maketitle
\begin{abstract}
 A counting constraint satisfaction problem (\#CSP) asks for the
 number of ways to satisfy a given list of constraints, drawn from a
 fixed constraint language $\Gamma$.  We study how hard it is
 to evaluate this number approximately.  There is an interesting
 partial classification, due to Dyer, Goldberg, Jalsenius and Richerby
 \cite{bdddeg}, of Boolean constraint languages when the \emph{degree}
 of instances is bounded by $d\geq 3$ - every variable appears in at most
 $d$ constraints - under the assumption that ``pinning'' is allowed as
 part of the instance. We study the $d=2$ case under the stronger
 assumption that ``variable weights'' are allowed as part of the
 instance.  We give a dichotomy: in each case, either the \#CSP is
 tractable, or one of two important open problems, $\BIS$ or $\PMP$,
 reduces to the \#CSP.
\end{abstract}
\clearpage

\input{intro}
\input{defns}
\input{reductions}
\input{sigtheory}
\input{basicthm}
\input{weighted}
\input{highdeg}
\input{fpras}

\subsection{Acknowledgements}

The proof of Theorem \ref{thm:basic} grew out of a
study of degree-two $\nCSP$s by Leslie Ann Goldberg and David
Richerby, and their discussions with the author.  The author wishes to
thank Leslie Ann Goldberg and Russell Martin for their advice.

\bibliography{main}{} \bibliographystyle{alpha}

\end{document}

%% file: intro.tex
\section{Introduction}

A constraint satisfaction problem asks whether there an assignment of
values to some variables that satisfies given constraints.  We will be
looking at Boolean CSPs, where each variable takes the value 0 or 1.
An example of a Boolean CSP is whether a graph has a perfect matching:
whether each edge can be labelled 0 or 1 (these are the variables)
such that (these are the constraints) at each vertex there is exactly
one edge labelled 1.

Given a finite set of relations $\Gamma$, the counting problem
$\nCSP(\Gamma)$ asks for the number of assignments that satisfy a
conjunction of constraints of of the form ``$(v_1,\cdots,v_k)\in R$''
with $R\in \Gamma$.  The approximation complexity of $\nCSP(\Gamma)$
is the complexity of the same problem but allowing a multiplicative
error.  Sometimes we will allow weighted constraints, called
\emph{signatures}, and in this case we write $\calF$ instead of $\Gamma$.

An important feature of the perfect matchings example is that every
variable is used twice: the \emph{degree} of every variable is two.
For larger degree bounds $\nCSP(\Gamma)$ has been studied in
\cite{bdddeg}.  The restriction of $\nCSP(\calF)$ to instances where
each variable appears exactly twice has also been called
a (non-bipartite) Holant problem \cite{holantc}.

To make progress on the degree two problem we allow instances to
specify a weight for each of the two values each variable can take.
The main result of the paper is a hardness result for degree two
Boolean $\nCSP$s with these variable weights: in every case we show
that that problem is either tractable or as hard as an important open
problem.
The core of the proof is that we can adapt the ``fan-out'' constructions
of Feder \cite{feder}; this does not work for delta matroids,
but delta matroids can be handled specially.  Along
the way we give a generalisation of delta matroids to
weighted constraints called ``terraced signatures''. This definition
directly describes when a constraint fails to give fan-out gadgets for
degree-two $\nCSP$s.

We also give partial results for signatures and for some
related problems.

\subsection{Variable weights and degree bounds}

We will consider the problem of approximately evaluating a $\nCSP$
where the constraints, variables weights, and degrees are restricted.
To discuss these problems it is useful to introduce some notation.
For the main theorem we study the problems $\nCSPs_{\leq 2}(\Gamma)$
for a constraint language $\Gamma$ of Boolean relations. The
instances of $\nCSPs_{\leq 2}(\Gamma)$ consist of variable weights and constraints.
Variable weights are arbitrary non-negative rationals, constraints are
taken from $\Gamma$, and every variable appears at most twice.

To discuss other results, and to put our results in a wider context,
it is useful to generalise from $\nCSPs(\Gamma)$.  Given a set of
non-negative ``variable weights'' $W\subset\mathbb{R}\times\mathbb{R}$
and a set of degree bounds $K\subseteq\mathbb{N}$, we then have an
approximate counting problem $\nCSP^W_K(\Gamma)$: instances consist of
a pair of variable weights from $W$ for each variable, and a set of
constraints from the set of relations $\Gamma$, such that the degree
of each variable is an integer in $K$.  To avoid clutter we will use
the {\bf default values} $W=\{(1,1)\}$ and $K=\mathbb{N}$ when they
are omitted, and {\bf abbreviate} $W=\Qnonneg\times\Qnonneg$ to $\geq
0$, and $K=\{1,\cdots,d\}$ and $K=\{d\}$ to $=d$ and $\leq d$
respectively.  We will in fact generalise to sets of signatures
$\calF$ and define $\nCSP^W_K(\calF)$. See Section \ref{ssec:cspdefn}
for more formal definitions.

For example, if we define $\NAND=\{(0,0),(0,1),(1,0)\}$, then
$\nCSP(\{\NAND\})$ is equivalent to the problem of counting
independent sets in a graph: the variables $x_v$ of the CSP correspond
to vertices $v$ of a (multi)graph, the constraints correspond to edges
- there is a constraint $\NAND(x_u,x_v)$ for each edge $uv$ of the
graph - and the satisfying assignments of the CSP are the indicator
functions of independent sets of this graph.
As another example, if we define $\PM_3=\{(0,0,1),(0,1,0),(1,0,0)\}$
then $\nCSP_{=2}(\{\PM_3\})$ is equivalent to counting perfect matchings
of a graph in which every vertex has degree three (by the same
encoding discussed previously for perfect matchings as a CSP),
and $\nCSPs_{=2}(\{\PM_3\})$ is equivalent to counting weighted perfect
matchings in a graph in which every vertex has degree three.

\subsection{Main result}

In approximation complexity a problem is considered tractable if it
has a fully polynomial randomised approximation scheme (FPRAS) - see
Section \ref{sec:fprasdefn} for a definition.  We will present results
using the ``AP-reduction'' notation $\APred$ introduced in
\cite{rcap}. $\Xprob\APred\Yprob$ means that $\Xprob$ has an FPRAS
using an FPRAS for $\Yprob$ as an oracle.  This also defines an
equivalence relation $\Xprob\APeq\Yprob$.

The main result states reductions from the problems $\SAT$, $\BIS$ and
$\PMP$ to certain $\nCSP$ problems. $\SAT$ is the problem of counting
solutions to a SAT instance; it plays a similar role for approximation
problems as NP plays for decision problems. $\BIS$ is the problem of
counting the number of independent sets in a bipartite graph. We do
not actually use this definition directly; $\BIS$ has been used in
this way as a ``hard'' problem since it was introduced in \cite{rcap}.
$\PMP$ is the problem of counting perfect matchings in a graph.
Finding an FPRAS for $\PMP$ has been an important open research
problem, certainly since the restriction of $\PMP$ to bipartite graphs
was shown to have an FPRAS \cite{allperm}.  It is therefore a
respectable ``hard'' problem for approximation.

We will give AP-reductions depending on whether $\Gamma$ falls into
certain classes of relations.  Briefly, a relation is basically binary
if it is a Cartesian product of relations of arity at most two, for
example $\{\vecx\in\{0,1\}^4\mid \text{$x_1x_2=1$ and $x_3\leq
  x_4$}\}$.  A relation is in NEQ-conj if it is a conjunction of
equalities, disequalities, and constants, for example
$\{\vecx\in\{0,1\}^6\mid x_1=x_2, x_2\neq x_5, x_6=0\}$.  A relation
is in IM-conj if it is a conjunction of implications and constants,
for example $\{\vecx\in\{0,1\}^6\mid x_1\leq x_2\leq x_3, x_6=0\}$.

A family $\calC$ of subsets of a finite set is a delta matroid if for
all $X,Y\in\calC$ and $i\in X\setdiff Y$ there exists $j\in X\setdiff
Y$ with $X\setdiff\{i,j\}\in\calC$, where the triangle operator means
the symmetric difference. In this paper we will also call the
corresponding relations $R\subseteq\{0,1\}^V$ delta matroids.  For
example, the set system $\{\emptyset,
\{1\},\{2\},\{1,2\},\{1,2,3\}\}\subset\{0,1\}^3$ is not a delta
matroid: it contains $X=\emptyset$ and $Y=\{1,2,3\}$ but does not
contain $\{3,j\}$ for any $j\in\{1,2,3\}$; hence the corresponding
relation $\{\vecx\in\{0,1\}^3\mid x_3\leq x_1,x_2\}$ is not a delta
matroid.  On the other hand $\{\vecx\in\{0,1\}^3\mid \sum
x_i\in\{0,2,3\}\}$ is a delta matroid relation.

Our main theorem says:

\begin{theorem}\label{thm:basic}
Let $\Gamma$ be a finite set of relations.  If $\Gamma\subseteq\Neqconj$ or
every relation in $\Gamma$ is basically binary then $\cwt(\Gamma)$ has an FPRAS.  Otherwise,
\begin{itemize}

\item If $\Gamma\subseteq\Imconj$ then $\BIS\APeq\cwt(\Gamma)$.
\item If $\Gamma\not\subseteq\Imconj$ then $\PMP\APred\cwt(\Gamma)$. If furthermore $\Gamma$ is not a set of delta matroids then $\SAT\APeq\cwt(\Gamma)$.
\end{itemize}
\end{theorem}

So in every case the problem is either tractable, or at least as hard as an
important open problem. This is quite a different situation from the
corresponding decision problems, considered in \cite{dalmauford}.  For
degree-two decision CSP there is no known dichotomy, and there
are many tractable problems using
delta matroids.

\subsection{Other results}

These classes or relations, and the proof of Theorem \ref{thm:basic},
generalises to some extent to signatures.
There is a similar notion of
basically binary signatures. NEQ-conj generalises to
Weighted-NEQ-conj, and IM-conj generalises to the class of
logsupermodular signatures (these classes were used in the result of
Bulatov et al. mentioned below).  We will define a generalisation of
delta matroids called ``terraced'' signatures.
We establish the following results in Section \ref{sec:weight}.

\begin{theorem}\label{thm:weightthm}
Let $\calF$ be a finite set of signatures.  If every signature in
$\calF$ is basically binary or every signature in $\calF$ is
in Weighted-NEQ-conj, then $\nCSPs_{\leq 2}(\calF)$ has an FPRAS.  Otherwise
assume furthermore that there is a signature in $\calF$ that is not
terraced or that does not have basically binary support. Then:
\begin{itemize}
\item If every signature in $\calF$ is logsupermodular then
  $\BIS\APred\nCSPs_{=2}(\calF)$.
\item If some signature in $\calF$ is not logsupermodular then
  $\PMp\APred\nCSPs_{=2}(\calF)$. If furthermore some
  signature in $\calF$ is not terraced then $\SAT\APeq\nCSPs_{=2}(\calF)$.
\end{itemize}
\end{theorem}

The case of terraced signatures whose support is basically binary is
left as an open problem.  Note that this theorem is stated for
$\nCSP_{=2}$ problems: every variable is used exactly twice and not
just at most twice.

\begin{theorem}\label{thm:finitewts}
Let $\calF$ be a finite set of signatures. Assume that not every
signature in $\calF$ is in Weighted-NEQ-conj, and not every signature
in $\calF$ is basically binary, and not every signature in $\calF$ is
terraced.
(This the same setting as the $\BIS$ and $\SAT$ reductions in Theorem
\ref{thm:weightthm}.)

Unless all the following conditions hold, there is a finite set
$W\subseteq\R_p\times\R_p$ such that
$\Xprob\APred\nCSP^W_{=2}(\calF)$ where $\Xprob=\BIS$ if every signature in
$\calF$ is logsupermodular, and $\Xprob=\SAT$ otherwise.

\begin{enumerate}
\item Every signature $F\in\calF$ is IM-terraced.
\item Either the support of every signature $F$ in $\calF$ is closed
  under meets
  ($\vecx,\vecy\in\supp(F)\implies\vecx\wedge\vecy\in\supp(F)$), or
  the support of every signature $F$ in $\calF$ is closed under joins
  ($\vecx,\vecy\in\supp(F)\implies\vecx\vee\vecy\in\supp(F)$).
\item No pinning of the support of a signature in $F$ is equivalent to
  $\EQ_2$.
\end{enumerate}
\end{theorem}

This situation is simpler for higher degrees, if $\calF$ contains a
signature with non-degenerate support (a relation is degenerate if it
is a product of arity 1 relations):

\begin{theorem}\label{thm:highdeg}
Let $\calF$ be a finite set of signatures and assume that not every
signature in $\calF$ has degenerate support. There exists a finite set
of variable weights $W$ such that $\nCSPs(\calF)$ has an FPRAS if and
only if $\nCSP_{\leq 3}^W(\calF)$ has an FPRAS.
\end{theorem}

So under these assumptions, by the theorem of Bulatov et al. mentioned below,
the tractable cases are just what can be computed exactly (unless \#BIS has an FPRAS). On the other hand, we show that the tractable region has positive measure, loosely speaking, for all $d\geq 2$:

\begin{theorem}\label{thm:boxg}
Let $d,k\geq 2$.  Let $F$ be a an arity $k$ signature with
values in the range $[1,\frac{d(k-1)+1}{d(k-1)-1})$.  Then
  $\nCSPs_{\leq d}(F)$ has an FPRAS.
\end{theorem}

\subsection{Related work}

The problem $\nCSP_{\leq d}(\Gamma)$ for $d\geq 3$ was studied in \cite{bdddeg}.
In particular:
\begin{theorem*}\cite[Theorem 24]{bdddeg}
Let $\Gamma$ be a finite set of relations and let $d\geq 6$.
\begin{itemize}
\item If every $R\in\Gamma$ is affine then
  $\nCSP_{\leq d}(\Gamma\cup\gpin)\in\text{FP}$.
\item Otherwise, if $\Gamma\subseteq\Imconj$ then
  $\nCSP_{\leq d}(\Gamma\cup\gpin)\APeq\BIS$.
\item Otherwise, there is no FPRAS for
  $\nCSP_{\leq d}(\Gamma\cup\gpin)$ unless NP=RP.
\end{itemize}
\end{theorem*}
Here $\Gamma_{\text{pin}}=\{\{(0)\},\{(1)\}\}$ and a relation is
called affine if it is an affine subspace of $\mathbb{F}_2^k$.

Theorem \ref{thm:weightthm} can be seen as an extension of the
following result of Bulatov et al \cite{lsm}, which we also rely on in
the proof:
\begin{lemma}\cite[Theorem 16]{lsm}\label{lem:unbounded}
\label{corthm:main}
Let $\calF$ be a finite set of signatures.  If $\calF$ is not a subset
of Weighted-NEQ-conj then for any finite subset $S$ of arity-one
signatures there is an FPRAS for $\nCSP(\calF\cup S)$.  Otherwise,
\begin{itemize}
\item there is a finite subset $S$ of arity-one signatures such that
$\BIS \APred \nCSP(\calF\cup S)$, and
\item if there is a function in $\calF$ that is not logsupermodular
then there is a finite subset $S$ of arity-one signatures such that
$\SAT \APeq \nCSP(\calF\cup S)$.
\end{itemize}
\end{lemma}

Note that arity one signatures are the same as variable weights for
unbounded degree $\nCSP$s. But when the degree is restricted, arity
one signatures seem less powerful.

Feder \cite{feder} showed that relations that are not delta matroids
give ``fan-out'': if $\Gamma$ contains a relation that is not a delta
matroid, and the decision problem $\operatorname{CSP}(\Gamma)$ is
NP-complete, then the restriction of $\operatorname{CSP}(\Gamma)$ to
degree two instances is also NP-complete.  Theorems \ref{thm:basic}
and \ref{thm:weightthm} use a similar kind of fan-out idea.  The
latest results on degree-two CSPs were given in
\cite{dalmauford}. There is no complete classification yet.

There are some important results on two-state spin systems that are
worth translating into the $\nCSP^W_K$ notation. As mentioned earlier,
Sly \cite{sly} showed that for the hard-core model there is a
computational transition at the ``tree threshold'' $\lambda^*(d)$,
$d\geq 6$ in the following sense.  It was known that the problem
$\nCSP^{\{(1,\lambda)\}}_{\leq d}(\NAND)$ has a (deterministic) FPRAS
for $\lambda<\lambda^*(d)$. Sly showed that it does not have an FPRAS
for $\lambda>\lambda^*(d)$ unless NP=RP (with some technical
restrictions on $\lambda$). This result has been extended recently
\cite{sly2} considering other models and removing the restrictions.
On the other hand there are FPRASes for variants of $\nCSP^W_K(B)$ for
various symmetric binary signatures $B$: see \cite{twospinbdd} for
$K=\{1,\cdots,d\}$ and \cite{twospin} for $K=\mathbb{N}$.

\def\hol{\operatorname{Holant}}
\def\holc{\operatorname{Holant}^c}
\def\hols{\operatorname{Holant}^*}

To discuss other work it is useful to define some notation
temporarily.  Define $\hol(\calF)=\nCSP_{=2}(\calF)$, and $\holc$ is
the same except that any arity one relation can be used, and $\hols$
is the same except that any arity one complex-valued signature can be
used.  $\hols$ was introduced in \cite{holantc} to give results about
the exact counting complexity (not allowing multiplicative error)
of $\holc$ problems.  A dichotomy theorem for the exact counting complexity of
$\hols$ problems was given in
\cite{holants}, classifying each problem as polynomial-time
computable or \#P-hard.

Yamakami \cite{tomo} studied the approximation complexity of
$\hols(\{F\})$ (referring to it as $\nCSP^*_2$) where $F$ is in a
certain set of arity three complex-valued signatures. It would be too
much of a detour to present those results fully, but the conclusion is
that these problems are either tractable or there is a certain
approximation-preserving reduction from the problem
$\SAT^*_{\mathbb{C}}$ (analogous to $\SAT$) to
$\operatorname{Holant}^*(\{F\})$.  Note that the node weight functions
used to define $\SAT^*_{\mathbb{C}}$ are like variable weights, but
the problems $\nCSP^*$ and $\hols$ defined in that paper do not use
variable weights, but arity one signatures.  In the same setting there
are results for higher degree bounds \cite{tomohigh}.

%% file: defns.tex
\section{Definitions}\label{sec:defns}
\def\mdef#1{{\bf #1}}

$V$ will usually denote a finite set whose elements are called
variables.  Elements of $\{0,1\}^V$ will be called
\mdef{configurations} of $V$.  In this paper a \mdef{relation} $R$ on
$V$ is a subset $R\subseteq\{0,1\}^V$.  In this paper a
\mdef{signature} $F$ on $V$ is a function $F:\{0,1\}^V\to\Rnonneg$,
where $\Rnonneg$ is the set of non-negative polynomial-time computable
reals, that is, non-negative reals $r$ for which there is a
polynomial-time Turing machine that when given an integer $n$ in
unary, outputs the first $n$ bits of the binary expansion of $r$.  The
set $V=V(R)=V(F)$ is called the variable set; the \emph{arity} is
$|V|$, and configurations in $\{0,1\}^k$ for integers $k$ are
considered to have variable set $\{1,\cdots,k\}$.

We can rename the variables in an obvious way. (For any finite set
$V'$, a bijection $\pi:V\to V'$ induces a bijection $\pi_*$ from
relations (or signatures) on $V$ to relations (or signatures) on
$V'$.)  We will say that relations (or signatures) are
\mdef{equivalent} if they are related by renaming variables.  The
difference between equivalent relations (or signatures) is never
important in this paper, but keeping track of $V$ makes some arguments
easier.

We will \emph{implicitly} convert relations to signatures, so
$R(\vecx)=1$ if $\vecx\in R$ and $R(\vecx)=0$ otherwise.  However, if
$R$ is given in set notation we will instead use the more legible
notation $\vecone_R(\vecx)=R(\vecx)$.

It is useful to have special notation for inverting components of a
configuration.  For all $\vecx\in\{0,1\}^V$ and all subsets
$U\subseteq V$ define the \mdef{flip} $\vecx^U\in\{0,1\}^V$ by
$\vecx^U_v\neq \vecx_v$ if and only if $v\in U$.  A relation $R$ or
signature $F$ can also be flipped: $\vecx\in R^U$ if and only if
$\vecx^U\in R$, and $F^U(\vecx)=F(\vecx^U)$.  Also, by abuse of
notation, for configurations $\vecx,\vecy\in\{0,1\}^V$, the set of
elements on which $\vecx$ and $\vecy$ differ will be denoted
$\vecx\setdiff\vecy$.

We will use $\veczero$ and $\vecone$ to mean the all-zero and all-one
configurations on some variable set.  The complement
$\overline{\vecx}$ of a configuration $\vecx$ is defined by
$\overline{x}_i=1-x_i$.  Define the \emph{meet} $\vecx\wedge\vecy$ and
\emph{join} $\vecx\vee\vecy$ of configurations
$\vecx,\vecy\in\{0,1\}^V$ by $(\vecx\wedge\vecy)_i=\min(x_i,y_i)$ and
$(\vecx\vee\vecy)_i=\max(x_i,y_i)$.

\subsection{Relations}
Let $R\subseteq\{0,1\}^V$ be a relation.  $R$ is an equality if it is
of the form $\{\vecx:x_i=x_j\}$.  $R$ is a disequality if it is of the
form $\{\vecx:x_i\neq x_j\}$.  $R$ is a pin if it is of the form
$\{\vecx:x_i=c\}$.  $R$ is an implication if it is of the form
$\{\vecx:x_i\leq x_j\}$.  Here $i,j\in V$ and $c\in\{0,1\}$.

 Define \mdef{NEQ-conj} to be the class of
relations that are conjunctions of equalities, disequalities, and
pins.  Define \mdef{IM-conj} to be the class of relations that are
conjunctions of implications and pins; we will often use the
characterisation that a relation is in IM-conj if and only if it is
closed under meets and joins (\cite[Corollary 18]{csptrich}).  $R$ is
a \mdef{delta matroid} if for all $\vecx,\vecy\in R$ and for all $i\in
\vecx\setdiff\vecy$ there exists $j\in \vecx\setdiff\vecy$, not
necessarily distinct from $i$, such that $\vecx^{\{i,j\}}\in R$.

A non-empty relation $R$ on a non-empty variable set is
\mdef{decomposable} if it is equivalent to the Cartesian product of at
least two relations of arity at least one. Otherwise it is
\mdef{indecomposable}. A relation is defined to be \mdef{degenerate}
if it is equivalent to the Cartesian product of relations of arity at
most one.  A relation is defined to be \mdef{basically binary} if it
is equivalent to the Cartesian product of relations of arity at most
two.

We will use the following relations.
$\EQ_k=\{\veczero,\vecone\}\subseteq\{0,1\}^k$,
$\NEQ=\{(0,1),(1,0)\}$, 
$\PIN_0=\{(0)\}$, $\PIN_1=\{(1)\}$,
$\NAND=\{(0,0),(0,1),(1,0)\}$,
$\OR=\{(0,1),(1,0),(1,1)\}$, and $\IMP=\{(0,0),(0,1),(1,1)\}$.
Also
$\PM_k=\{\vecx\in\{0,1\}^k\mid x_1+\cdots+x_k=1\}$.

\subsection{Pinnings}
A \mdef{partial configuration} $\vecp$ of $V$ is defined to be an
element of $\{0,1\}^{\dom(\vecp)}$ for some subset
$\dom(\vecp)\subseteq V$.  If
$\vecx\in\{0,1\}^{V\setminus\dom(\vecp)}$ then $(\vecx,\vecp)$ means
the unique common extension of $\vecx$ and $\vecp$ to a configuration
of $V$.  Let $R\subseteq\{0,1\}^V$ and let $\vecp$ be a partial
configuration of $V$.  Define the (relation) \mdef{pinning} $\pin R
\vecp\subseteq\{0,1\}^{V\setminus \dom(\vecp)}$ by $\vecx\in\pin R
\vecp \iff (\vecx,\vecp)\in R$.  Let $F\sigf V$ and let $\vecp$ be a
partial configuration of $V$.  Define the (signature) \mdef{pinning}
$\pin F \vecp\sigf {V\setminus \dom(\vecp)}$ by $\pin F
\vecp(\vecx)=F(\vecx,\vecp)$.  In the delta matroid literature, the
set system representation of a pinning is called a minor.

\subsection{Signatures}
Let $V$ and $V'$ be finite sets, and let $V\sqcup V'$ be their
disjoint union.  The \mdef{tensor product} $F\otimes G\sigf{V\sqcup
  V'}$ of two signatures $F\sigf V$ and $G\sigf {V'}$ is defined by
$(F\otimes G)(\vecx,\vecx')=F(\vecx)G(\vecx')$ for all
$\vecx\in\{0,1\}^V$ and $\vecx'\in\{0,1\}^{V'}$.  We can define the
tensor product of $m$ signatures $\bigotimes_{i=1}^m F_i=
F_1\otimes(F_2\otimes \cdots\otimes(F_{m-1}\otimes F_m)\cdots)$.  A
signature is \mdef{decomposable} if it is equivalent to a tensor
product of signatures of arity at least one.  Otherwise it is
\mdef{indecomposable}.

A signature is defined to be \mdef{degenerate} if
it is equivalent to the tensor product of two signatures of arity one.  A
signature is defined to be \mdef{basically binary} if it is equivalent
to the product of signatures of arity at most two.

$F'$ is a \mdef{simple weighting} of $F$ if $F'$ is the pointwise
product $F'(\vecx)=F(\vecx)D(\vecx)$ of $F$ with a degenerate
signature $D$.  Define \mdef{Weighted-NEQ-conj} to be the class of
simple weightings of NEQ-conj relations - see Proposition
\ref{prop:wncisclone} for how this related to Lemma
\ref{lem:unbounded}.  A signature $F\sigf V$ is \mdef{logsupermodular}
if it satisfies $F(\vecx\wedge\vecy)F(\vecx\vee\vecy)\geq
F(\vecx)F(\vecy)$ for all $\vecx,\vecy\in\{0,1\}^V$.

We now come to the definition of terraced signatures, which are
signatures such that the reductions in Section \ref{sec:weight}
(ultimately Lemma \ref{lem:holanty}) fail.  In Lemma \ref{lem:dmterr}
we will show that a relation is terraced if and only if it is a delta
matroid, so we are defining a weighted generalisation of delta
matroids.

A signature $F\sigf V$ is \mdef{terraced} if for all partial
configurations $\vecp$ of $V$ and all $i,j$ in the domain of $\vecp$,
if $\pin F \vecp$ is identically zero then $\pin F {\vecp^{\{i\}}}$
and $\pin F {\vecp^{\{j\}}}$ are linearly dependent, that is, one is a
scalar multiple of the other. The scalars can depend on $i$ and $j$.
A signature $F\sigf V$ is \mdef{IM-terraced} if for all partial
configurations $\vecp$ of $V$ and all $i,j$ in the domain of $\vecp$
such that $p_i\neq p_j$, if $\pin F \vecp$ is identically zero then
$\pin F {\vecp^{\{i\}}}$ and $\pin F {\vecp^{\{j\}}}$ are linearly
dependent.

Let $V$ be a finite set, let $F\sigf V$ and let $h:V\to\mathbb{Z}$.
Define the \mdef{$h$-maximisation} $F_{h-\max}\sigf V$ by setting
$F_{h-\max}(\vecx)=F(\vecx)$ for all configurations $\vecx$ of $V$
such that $\sum_ix_ih_i=\max_{\vecy\in\supp(F)}\sum y_ih_i$, and setting
$F_{h-\max}(\vecx)=0$ otherwise.

\subsection{K-formulas}
Our $\nCSP$ instances will use a ``primitive product summation (pps)''
formula as in \cite{lsm}.  These can be thought of as formal
summations of products of function applications such as $\sum_y
\NEQ(x,y)\NEQ(y,z)$.

For a set of signatures $\calF$, a \mdef{pps-formula $\phi$ over
  $\calF$} consists of an external variable set $V=V^\phi$, an
internal variable set $U=U^\phi$ disjoint from $V$, a set of atomic
formula indices $I=I^\phi$, a signature $F_i=F^\phi_i\in\calF$ for
each $i\in I$, and scope variables $\scope(i,j)=\scope^\phi(i,j)\in
U\cup V$ for each $i\in I$ and $j\in V(F_i)$.  The data associated to
an index $i\in I$ ($F_i$ and $\scope(i,j)$ for $j\in V(F_i)$) is
called an atomic formula, denoted by a formal function application
like $F_i(v_1,v_2,v_3)$. We will manipulate pps-formulas by
\emph{inserting} or \emph{deleting} atomic formulas to obtain a new
pps-formula.

Define $\zp{\phi}\sigf V$ as follows: for all configurations
$\vecx$ of $V$,
\[\zp{\phi}(\vecx)=\sum \prod_{i\in I} F_i((x_{\scope(i,j)})_{j\in V(F_i)})\]
The sum is over all extensions of $\vecx$ to a configuration of $U\cup
V$, and the notation $(x_{\scope(i,j)})_{j\in V(F_i)}$ means the
configuration in $\{0,1\}^{V(F_i)}$ given by the composition
$V(F_i)\xrightarrow{\scope(i,\bullet)} V\xrightarrow{\vecx}\{0,1\}$.

This gives a quick way to specify all the data.  The \emph{pps-formula
  given by
\[ \zp{\phi}(x_1,\cdots,x_n)=\sum_{x_{n+1},\cdots,x_{n+m}} \prod_{i\in I} F_i(x_{\scope(i,1)},\cdots,x_{\scope(i,a_i)}) \label{eqn:ppsimplicit}\]
for all $x_1,\cdots,x_n\in\{0,1\}$,} is the pps-formula with
$V=\{1,\cdots,n\}$ and $U={n+1,\cdots,n+m}$ and the given $I,F_i$ and
$\scope$.  (For this to make sense we must have
$V(F_i)=\{1,\cdots,a_i\}$ for each $i\in I$, and the $\scope(i,j)$
values must fall in $\{1,\cdots,n+m\}$.)  We will say a signature $G$
\emph{is defined by a pps-formula over $\calF$} if $G=\zp{\phi}$ for some
pps-formula $\phi$ over $\calF$. The variables do not have to be
called $x_1,\cdots,x_{n+m}$; for example we could say that $\EQ_2$ is
defined by a pps-formula over $\{\NEQ\}$ because
\[ \EQ_2(x,z)=\sum_y \NEQ(x,y)\NEQ(y,z) \]
for all $x,z\in\{0,1\}$.

The degree $\deg_\phi(v)$ of an internal or external variable $v\in
U\cup V$ is the number of times it occurs: the number of pairs $(i,j)$
such that $\scope(i,j)=v$.  For any subset $K$ of natural numbers, a
$K$-formula is a pps-formula where if
\footnote{if $K=\mathbb{N}$ then degrees do not matter, so we allow
  any pps-formula and do not insist that the external variables have
  degree 1} $K\neq \mathbb{N}$ then: the degree of every internal
variable is in $K$, and the degree of every external variable is $1$.
($\leq d$)-formulas and ($=d$)-formulas are $K$-formulas with
$K=\{1,\cdots,d\}$ and $K=\{d\}$ respectively.  As above we can say
\emph{the K-formula given by} some equation of the form
\eqref{eqn:ppsimplicit}, and we can say a signature is \emph{defined
  by a K-formula over $\calF$}.\footnote{This is similar to to
  ``realizing'' a signature in \cite{holantc}, and T-constructibility
  in \cite{tomo}.}

\begin{proposition}\label{prop:wncisclone}
A signature $F\sigf k$ is in Weighted-NEQ-conj if and only if
$F=\zp{\phi}$ for some pps-formula using $\EQ_2$, $\NEQ$ and arity 1
signatures.  Hence the version of Lemma \ref{lem:unbounded} given in
the introduction is a faithful translation.
\end{proposition}
\begin{proof}
For the forward direction it is easy to construct such a formula
$\phi$. For the backward direction it will be convenient to first note
a few properties of Weighted-NEQ-conj. In an indecomposable NEQ-conj
relation $R$, every two variables are related by a chain of equalities
and disequalities, so $R\subseteq\{\vecx,\overline{\vecx}\}$ for some
$\vecx$.  An indecomposable signature in Weighted-NEQ-conj must have
indecomposable support, so an indecomposable signature in
Weighted-NEQ-conj has support of cardinality at most two.

Conversely, it is easy to check that any relation of cardinality at
most two is in NEQ-conj, and any signature $F$ whose support has
cardinality at most two is in Weighted-NEQ-conj.  We can now check
each stage of the expression for $\zp{\phi}$:
(1.) If $F$ is in Weighted-NEQ-conj then so is $F'(\vecx)=F((x_{\scope(i,j)})_{j\in V(F)})$.
(2.) If two signatures are in Weighted-NEQ-conj then so is their pointwise product.
(3.) If $F(t,\vecx)$ is in Weighted-NEQ-conj then so is $F'(\vecx)=\sum_t F(t,\vecx)$.
The first two stages are obvious from the definition of
Weighted-NEQ-conj.  For the third stage, note that Weighted-NEQ-conj
is closed under tensor products so we can assume that $F$ is
indecomposable. Then $|\supp(F')|\leq|\supp(F)|\leq 2$ so $F'$ is in
Weighted-NEQ-conj.
\end{proof}

\subsection{\#CSPs}\label{ssec:cspdefn}
We will now formalise the definitions given in the introduction.

We will call $W$ a \emph{set of variable weights} if one of the
following conditions holds.
\begin{itemize}
\item $W\subseteq\Qnonneg \times \Qnonneg$; elements of $W$ will be specified
  as binary fractions.  The binary representation is important - see
  Section \ref{sec:monodim}.
\item $W$ is a finite subset of $\Rnonneg\times\Rnonneg$; elements
  of $W$ will be specified by their index in a fixed enumeration.
\end{itemize}
Let $\calF$ be a finite set of signatures, let $W$ be a set of
variable weights and let $K$ be a set of positive integers.  A
$\nCSP_K^W(\calF)$ instance $(w,\phi)$ consists of
a function $w:V\to W$, and a $K$-formula $\phi$ with no external
variables and with internal variables $V$, where $V=V^\phi$. The value of the instance
is
\[\zwp{w}{\phi}=\sum_{\vecx:V\to\{0,1\}}\left(\prod_{v\in V} w(v)_{x_v}\right)\left(\prod_{i\in I}F_i((x_{\scope(i,j)})_{j\in V(F_i)})\right)\]
where the $I,F_i,\scope$ are given by $\phi$. If $W=\{(1,1)\}$ we will
omit $w$, so the instance is $\phi$ and the output is $\zp{\phi}$ (a
slight abuse of notation - here $\zp{\phi}$ means the value of
$\zp{\phi}$ applied to the arity zero configuration).  It will
occasionally be useful to refer to the contribution
$\wt{w}{\phi}(\vecx)$ of a configuration $\vecx$:
\[\wt{w}{\phi}(\vecx)=\left(\prod_{v\in V} w(v)_{x_v}\right)\left( \prod_{i\in I}F_i((x_{\scope(i,j)})_{j\in
  V(F_i)})\right)\]

\subsection{Approximation complexity}\label{sec:fprasdefn}

The paper \cite{rcap} introduced an analogue of Turing reductions for
approximation problems, which we repeat here (except that we generalise
by allowing $f$ to take non-integer values, as in \cite{lsm}).

A \emph{randomised approximation scheme} for a function $f:\Sigma^* \to
\Rnonneg$ is a probabilistic Turing machine (TM) that takes as input a pair
$(x,\epsilon)\in\Sigma^*\times(0,1)$ and produces as output an rational
random variable $Y$ satisfying the condition
$\text{Pr}(\exp(-\epsilon)f(x) \leq Y \leq \exp(\epsilon) f(x))\geq 3/4$.  A
randomised approximation scheme is said to be fully polynomial if it
runs in time $\text{poly}(|x|, \epsilon^{-1})$.  The phrase ``fully
polynomial randomised approximation scheme'' is usually abbreviated to
FPRAS.

Let $f,g : \Sigma^* \to \Rnonneg$ be functions whose complexity (of
approximation) we want to compare. An \emph{approximation-preserving
  reduction from $f$ to $g$} is a probabilistic
oracle TM $M$ that takes as input a pair $(x,\epsilon) \in \Sigma^*
\times (0, 1)$, and satisfies the following three conditions: (i)
every oracle call made by $M$ is of the form $(w,\delta)$, where $w
\in \Sigma^*$ is an instance of $g$, and $0<\delta<1$ is an error
bound satisfying $\delta^{-1}\leq \text{poly}(|x|, \epsilon^{-1})$;
(ii) the TM $M$ meets the specification for being a randomised
approximation scheme for $f$ whenever the oracle meets the
specification for being a randomised approximation scheme for $g$; and
(iii) the run-time of $M$ is polynomial in $|x|$ and
$\epsilon^{-1}$. If an approximation-preserving reduction from $f$ to
$g$ exists we write $f\APred g$, and say that f is AP-reducible to
$g$. If $f\APred g$ and $g\APred f$ then we write $f \APeq g$.

%% file: reductions.tex
\section{Reductions}\label{sec:reductions}
This section establishes some reductions between $\nCSP$s.

We will often implicitly use the fact that
$\nCSP_K^W(\calF)\APred\nCSP_{K'}^{W'}(\calF)$ whenever $K\subseteq
K'$ and either $W\subseteq W'$ or $W'=\Qnonneg\times\Qnonneg$.  The
reduction is trivial except in the case where $W$ consists of a finite
set of polynomial-time computable variable weights and
$W'=\Qnonneg\times\Qnonneg$; in this case the reduction just needs to
choose good enough approximations to the variable weights in $W$.

$K$-formulas are designed to be used as gadgets in the following sense.

\begin{lemma}\label{lem:express}
Let $\calF$ be a finite set of signatures.  Let $W$ be a
set of variable weights containing $(1,1)$.  Let
$K\subseteq\mathbb{N}$.  Let $\psi$ be a $K$-formula. Then
$\nCSP^W_K(\calF\cup\{\zp{\psi}\})\APred\nCSP^W_K(\calF)$.
\end{lemma}
\begin{proof}
Given an instance $(w,\phi)$ of $\nCSP(\calF\cup\{\zp{\psi}\})$, for
each atomic formula $\zp{\psi}(s)$, delete that atomic formula and
insert a copy of each atomic formula in $\psi$, renaming the external
variables $v\in V(\psi)$ of $\psi$ to $s(v)$ and renaming the internal
variables of $\psi$ to fresh variables.  This process gives a new
instance $(w',\phi')$ over $\calF$ on a possibly larger variable set
$V'$, where we extend $w$ to $w'$ by setting $w'(v,0)=w'(v,1)=1$ for
all new variables $v$.

In terms of $\zwp{w}{\phi}$, this process has the effect of replacing
each use of $\zp{\psi}$ by its summation-of-product definition and
distributing out the sums over the internal variables. By
distributivity $\zwp{w}{\phi}=\zwp{w'}{\phi'}$, and the degrees are
all still in $K$ so we can call the oracle on $\zwp{w'}{\phi'}$
without changing the error parameter $\epsilon$.
\end{proof}

The following reduction is an important step in the proof of Theorem
\ref{thm:basic}: it shows that we can get $\PM_3$ from
$\{(x_1,x_2,x_3)\in\{0,1\}^3\mid x_1+x_2+x_3\leq 1\}$ for example, unlike in the
finite $W$ setting of Section \ref{sec:monodim}.

\begin{lemma}\label{lem:hamweight}
Let $\calF$ be a finite subset of signatures.  Let $G\in\calF$ and let
$h:V\to\mathbb{Z}$ where $V$ is the variable set of $G$.  Let
$W=\Qnonneg\times\Qnonneg$ (we will also allow $W=\{(2^a,2^b)\mid
a,b\in\mathbb{Z}\}$ for the proof of Theorem \ref{thm:finitewts}).
Then
\[\nCSP^W_K(\calF\cup\{G_{h-\max}\})\APred\nCSP^W_K(\calF)\]
\end{lemma}
\begin{proof}

The reduction is given an instance $(w,\phi)$ of
$\nCSP^W_K(\calF\cup\{F_{h-\max}\})$ and error parameter $\epsilon$
which we can assume is less than $1/2$.  We wish to compute a value
$Z$ such that $\exp(-\epsilon)Z\leq \zwp{w}{\phi}\leq\exp(\epsilon)Z$.

Let $s=|V|+|I^\phi|$ be the total number of variables and atomic
formulas in $\phi$. Let $M$ be the maximum over: the values taken by
signatures in $\calF$, and the values $w(v,i)$, and the value $1$.
Let $m$ be the minimum over: the non-zero values taken by signatures
in $\calF$, and the non-zero values $w(v,i)$, and the value $1$.  Let
$H$ be the maximum of $\sum_ix_ih_i$ over $\vecx\in\supp(F)$. Define
$G^{(n)}$ for all $n\geq 0$ by
\[ G^{(n)}(\vecx)=G(\vecx)2^{n(\sum_ix_ih_i-H)} \]
Note that for all $\vecx$ either: $\sum_ix_ih_i=H$ so
$G^{(n)}(\vecx)=G_{h-\max}(\vecx)$, or $\sum_ix_ih_i<H$ so
$G_{h-\max}(\vecx)=0$ and $G^{(n)}(\vecx)\leq M2^{-n}$.

 Let $n=\lceil |V|+s\log M-\log_2(m^s\epsilon/4) \rceil$.  Note that
 $2^{|V|+s\log M-n}\leq (\epsilon/4)m^s$.  Let $I'$ be the set of
 atomic formula indices such that $F^\phi_i=G_{h-\max}$.  Let $\phi'$ be the
 same as $\phi$ except that $F^{\phi'}_i=G^{(n)}$ for each $i\in I'$.

Let $Z=\zwp{w}{\phi}$ and $Z'=\zwp{w}{\phi'}$.
We can approximate $Z'$ using the oracle by replacing
$G^{(n)}$ by variable weights and $G$.  Specifically, let $\phi''$ be
the same as $\phi$ except that $F^{\phi''}_i=G$ for each $i\in I'$.
  For each variable $v$
let $h(v)$ be the sum of $h_j$ over all $i\in I'$ and $j\in V(G)$ such
that $\scope^\phi(i,j)=v$.  Let $w''(v,0)=w(v,0)$ and
$w''(v,1)=w(v,1)2^{nh(v)}$.  Then
$\zwp{w''}{\phi''}=\zwp{w}{\phi'}2^{nH|I'|}$. Call the oracle on
$(w'',\phi'')$ with error parameter $\epsilon/2$ and divide the result
by $2^{nH|I'|}$ to obtain a value $Z''$ such that
$\exp(-\epsilon/2)Z'\leq Z''\leq\exp(\epsilon/2)Z'$ with probability
at least $3/4$.

For all configurations $\vecx$, if $\wt{w}{\phi}(\vecx)\neq
\wt{w}{\phi'}(\vecx)$ then $\wt{w}{\phi}(\vecx)=0$ and
$\wt{w}{\phi'}(\vecx)\leq M^s2^{-n}$. Hence
\begin{align*}
\left|Z-Z'\right|\leq 2^{|V|+s\log M-n}\leq m^s(\epsilon/4)
\end{align*}

If $Z\neq 0$ then $Z'' > Z'/2 > Z/4 \geq m^s/4 $.  The reduction can
therefore output zero whenever $Z''\leq m^s/4$.  If $Z=0$ then
$Z''\leq 2Z'\leq m^s(\epsilon/2) < m^s/4$ (for $\epsilon<1/2$).  So if
$Z''> m^s/4$ then we can assume $Z\neq 0$.  In this case we have
$|Z-Z'|\leq Z(\epsilon/4)$. Since $e^{-\epsilon/2}\leq 1-\epsilon/4$
for $\epsilon<2$,
\begin{align*}
(1-\epsilon/4)Z \leq Z'\leq (1+\epsilon/4)Z\\
  \exp(-\epsilon/2)Z \leq Z'\leq \exp(\epsilon/2)Z\\
  \exp(-\epsilon)Z \leq Z''\leq \exp(\epsilon)Z
\end{align*}
In this case the reduction can output $Z''$.
\end{proof}

Known polynomial-time algorithms can easily be modified to allow
variable weights:

\begin{lemma}\cite[Theorem 16]{lsm}, \cite[Theorem 2.2]{holants}\label{lem:tract}
Let $\calF$ be a finite set of signatures.  If $\calF$ is contained in
Weighted-NEQ-conj then $\nCSPs(\calF)$ has an FPRAS. If every signature
in $\calF$ is basically binary, then $\nCSPs_{\leq 2}(\calF)$ has an
FPRAS. In fact these problems are in FP, at least if the signatures in
$\calF$ take rational values.
\end{lemma}

The following Lemma is useful for showing that a problem is
AP-equivalent to $\SAT$.
\begin{lemma}\label{lem:sateasy}
Let $\calF$ be a finite set of signatures.  Then
$\nCSP^W_K(\calF)\APred\SAT$.
\end{lemma}
\begin{proof}
We can approximate the values in the signatures and variables weights
by rationals, and by scaling we can assume the values are in fact
integers. The problem of evaluating a $\nCSP$, with explicit
integer-valued signatures as part of the input, is in \#P and hence
AP-reduces to $\SAT$ - see the remarks in Section 3 of \cite{rcap}.
\end{proof}

We will use pinning throughout. The following Lemma shows that we do
not need to assume that $\PIN_0,\PIN_1$ are part of the constraint
language.
\begin{lemma}\label{lem:freepin}
Let $K$ be any non-empty set.  Let $\calF$ be a finite set of
signatures. Let $W$ be a set of variable weights containing $(1,0)$
and $(0,1)$. Then
\[ \nCSP_K^W(\calF'\cup\{\PIN_0,\PIN_1\})\APred\nCSP_K^W(\calF) \]
where $\calF'$ in the set of pinnings of signatures in $\calF$.
\end{lemma}
\begin{proof}
Let $G_0\in\calF$ be a signature with
$\supp(G_0)\not\subseteq\{\vecone\}$ and let $G_1\in\calF$ be a
signature with $\supp(G_1)\not\subseteq\{\veczero\}$.  (If these do not
exist then $\nCSP_K^W(\calF'\cup\{\PIN_0,\PIN_1\})$ has an FPRAS).

First we will establish that there is a $K$-formula $\psi$ over
$\calF$, of some arity $d$, such that $\veczero\in Z_\psi$.
Indeed there exists $\vecz\in\supp(G_0)$ and $i\in V(G_0)$  such that $z_i=0$.
We may assume $i=1$ and $V(G_0)=\{1,\cdots,k\}$ for some $k$.
Then pick $d\in K$ and let $\psi$ be the $\{d\}$-formula defined by
\[ Z_\psi(x_1,\cdots,x_d)=\sum_{y_2,\cdots,y_k} \prod_{i=1}^d G(x_i,y_2,\cdots,y_k) \]
By choice of $\vecz$ we have $\veczero\in Z_\psi$.

We will first show that
\[\nCSP_K^W(\calF\cup\{\PIN_0\}) \APred \nCSP_K^W(\calF\cup\{Z_\psi\})  \]
The reduction is given an instance $(w,\phi)$ of
$\nCSP_K^W(\calF\cup\{\PIN_0\})$.  By scaling - keeping track of an
overall multiplicative constant - we can assume that if there is an
atomic formula $\PIN_0(v)$ in $\phi$ then $w(v,0)=1$ and
$w(v,1)=0$. Take $d$ copies of this instance, but for each atomic
formula $\PIN_0(v)$ in $\phi$, rather than taking its $d$ copies
$\PIN_0(v_1)\cdots\PIN_0(v_d)$, insert the atomic formula
$Z_\phi(v_1,\cdots,v_d)$ where the scope consists of the $d$ copies of
$v$.  This process gives an instance $(w',\phi')$ of
$\nCSP_K^W(\calF\cup\{Z_\psi\})$.  Let $s$ be the number of $\PIN_0$
atomic formulas in $\phi$.  Then
$\zwp{w'}{\phi'}=Z_\psi(\veczero)^s(\zwp{w}{\phi})^d$.  So we get an
approximation to $\zwp{w}{\phi}$ within ratio $e^{\epsilon}$ by asking
the oracle for an approximation to $\zwp{w'}{\phi'}$ to within ratio
$e^{d\epsilon}$.

Using Lemma \ref{lem:express}, and by a symmetric argument to get $\PIN_1$, we have
\[\nCSP_K^W(\calF\cup\{\PIN_0,\PIN_1\}) \APred \nCSP_K^W(\calF)  \]
Pinnings can be expressed as $K$-formulas
using $\{\PIN_0,\PIN_1\}$, so again by Lemma \ref{lem:express}
\[ \nCSP_K^W(\calF'\cup\{\PIN_0,\PIN_1\})\APred\nCSP_K^W(\calF) \]

\end{proof}

When dealing with finite sets of variables weights in Theorem
\ref{thm:finitewts} it will be useful to be able to assume
$W=\{(1,1)\}$.  The following Lemma is not used in the proof of
Theorem \ref{thm:basic}, however.
\begin{lemma}\label{lem:simpleweight}
Let $K$ be a finite non-empty set of integers. Let $\calF$ be a finite set of signatures.
\begin{enumerate}
\item
Let $\calG$ be a finite set of simple weightings of signatures in $\calF$.
There is a finite set of variable weights $W$ such that
$ \nCSP_K(\calG)\APred\nCSP_K^W(\calF) $.
\item
For all finite sets of variable weights $W$ there
is a finite set $\calG$ of simple weightings of signatures in $\calF$
such that
$ \nCSP_K^W(\calF)\APred\nCSP_K(\calG) $.
\end{enumerate}
\end{lemma}
\begin{proof}
(1.)  Each $G\in\calG$ can be expressed as
  $G(\vecx)=F_G(\vecx)\prod_{j\in V(F)}U_{G,j}(x_j)$ for some
  $F_G\in\calF$ and arity 1 signatures $U_{G,j}$.  From now on we will
  let $G$ range over $\calG$ and $j$ range over $V(G)$.  Let $W$ be
  the set consisting of variable weights $(\prod_{G,j}
  U_{G,j}(0)^{n(G,j)},\prod_{G,j} U(1)^{n(G,j)})$ for all choices of
  $0\leq n(G,j)\leq \max(K)$. Then $|W|\leq
  (\max(K)+1)^{\sum_{G\in\calG}|V(G)|}$ so $W$ is a finite set.

Given an instance $\phi$ of $\nCSP_K(\calG)$, for each $G\in\calG$ and
each atomic formula $G(s)$, delete that atomic formula and insert an
atomic formula $F_G(s)$.  Define $w:V\to W$ by $w(v)_x=\prod_{G,j}
U_{G,j}(x)^{n(v,G,j)}$ for $x=0,1$, where $n(v,G,j)$ is the number of
  atomic formulas $G(s)$ of $\phi$ such that $s(j)=v$.  It follows
  that $\zp{\phi}=\zp{\phi'}$, so the reduction can just query the
  oracle with $(w',\phi')$, passing the instance's error parameter to
  the oracle.

(2.) Let $\calG$ consist of all signatures of the form
  $F(\vecx)\prod_{i\in V(F)}w(i)_{x_i}$ with $F\in\calF$ and
  $w:V(F)\to W\cup\{(1,1)\}$.

Given an instance $(w,\phi)$ of $\nCSP_K^W(\calF)$, let $V'=\{v\in
V\mid \deg(v)>0\}$. In terms of $Z_\phi$, we will regroup the factors
of $w(v)_{x_v}$ into an existing atomic formula, for each $v\in V'$.
Specifically, let $g:V'\to I^\phi$ be any map taking each variable $v\in
V'$ to the index of an atomic formula with $v$ in its scope:
$\scope^\phi(g(v),t)=v$ for some $t$.  Let $\phi'$ be the $K$-formula
with the same variables and scopes as $\phi$, but for each $i\in
I^\phi$, define $F^{\phi'}_i$ by
$F^{\phi'}_i(\vecx)=F(\vecx)\prod_{j\in V(F)}U_{i,j}(x_j)$ where
$U_{i,j}(y)=w(\scope(i,j),y)$ if $g(\scope(i,j))=i$, and
$U_{i,j}(y)=1$ otherwise.  Then
$\wt{w}{\phi}(\vecx)=\wt{}{\phi'}(\vecx)c$ for all configurations
$\vecx$ of $V$ where $C=\prod_{v\in V\setminus V'}(w(v,0)+w(v,1))$.
Thus $\zwp{w}{\phi}=C\zp{\phi'}$; the reduction can call the
$\nCSP_K(\calG)$ oracle to get an approximation to $\zp{\phi'}$ then
multiply by $C$.
\end{proof}

%% file: sigtheory.tex
\section{Minimal pinnings}
We will characterise various classes of signatures in terms of
pinnings.  This is in the same spirit as the ppp-definability studied
in \cite{bdddeg}.

For a class $P$ of relations, we will say a relation $R$ is
\emph{pinning-minimal $P$}, or \emph{pinning-minimal subject to $P$},
if $R$ is in $P$ and $R_\vecp$ is not in $P$ for any non-trivial
partial configuration $\vecp$.  Similarly we can say a signature is
pinning-minimal $P$.

Define a signature pair to be a pair $(F,G)$ of signatures $F,G\sigf
V$ for some $V$.  For a class $P$ of signature pairs we will say
$(F,G)$ is pinning-minimal $P$ if $(F,G)$ is in $P$ and $(\pin F
\vecp,\pin G \vecp)$ is not in $P$, for any non-trivial partial
configuration $\vecp$.  A signature pair $(F,G)$ is defined to be
linearly dependent if there exist $\lambda,\mu\in\R$, not both zero,
such that $\lambda F=\mu G$.

\begin{lemma}\label{lem:minlindep}
Let $(F,G)$ be a pinning-minimal linearly independent signature pair.
Then $\supp(F)\cup\supp(G)=\{\vecx,\overline{\vecx}\}$ for some
configuration $\vecx\in\supp(F)$.
\end{lemma}
\begin{proof}
First we give another characterisation of linear independence of a
signature pair.  For any $F',G'\sigf V$ consider the two-by-$2^{|V|}$
matrix $M$, with columns indexed by $\{0,1\}^V$, defined by
$M_{1,\vecx}=F(\vecx)$ and $M_{2,\vecx}=G(\vecx)$.  The signature pair
is linearly independent if and only if $M$ has row rank two, hence if
and only if $M$ has column rank two, and hence if and only if there
exist $\vecx,\vecy$ such that the two-by-two submatrix
\[M(\vecx,\vecy)=\begin{pmatrix}F(\vecx)&F(\vecy)\\G(\vecx)& G(\vecy)
\end{pmatrix}\]
has linearly independent rows.

Now let $(F,G)$ be a pinning-minimal linearly independent signature
pair.  For any $(\vecx,\vecy)$ such that $M(\vecx,\vecy)$ has linearly
independent rows, let $\vecp=\{i\mapsto x_i\mid x_i=y_i\}$. Then
$(\pin {F} \vecp,\pin {G} \vecp)$ is a linearly independent signature
pair. Hence $\vecy=\overline{\vecx}$.

There exists some $\vecx$ such that $M(\vecx,\overline{\vecx})$ has
linearly independent rows. For all $\vecz$ such that $F(\vecz)$ or
$G(\vecz)$ is non-zero either $M(\vecx,\vecz)$ has linearly
independent rows or $M(\vecz,\overline{\vecx})$ has linearly
independent rows. By the previous paragraph, $\vecz=\overline{\vecx}$
or $\vecz=\overline{\overline{\vecx}}=\vecx$.  Hence
  $\supp(F)\cup\supp(G)\subseteq\{\vecx,\overline{\vecx}\}$. Finally,
  since $F$ is not identically zero, one of $\vecx$ or
  $\overline{\vecx}$ is in $\supp(F)$
\end{proof}

\begin{lemma}\label{lem:minnonlsm}
Let $F$ be a pinning-minimal non-logsupermodular signature.  Then
$\supp(F)\subseteq\{\veczero,\vecx,\overline{\vecx},\vecone\}$ for
some $\vecx$.
\end{lemma}
\begin{proof}
For all $\vecx,\vecy$ such that
$F(\vecx\wedge\vecy)F(\vecx\vee\vecy)<F(\vecx)F(\vecy)$, the pinning
of $F$ by $\{i\mapsto x_i\mid x_i=y_i\}$ is not logsupermodular so
$\vecy=\overline{\vecx}$.  There exists such a tuple $\vecx$.  Also,
taking the contrapositive, for all $\vecy,\vecz$ such that
$\vecz\neq\vecy$ we have $F(\vecz\wedge\vecy)F(\vecz\vee\vecy)\geq
F(\vecz)F(\vecy)$.  Let $\vecz\notin\{\mathbf 0,\mathbf
1,\vecx,\overline{\vecx}\}$. Then
\begin{align*}
F(\vecx\wedge\vecz)F(\vecx\vee\vecz)\geq F(\vecx)F(\vecz)\\
F(\overline{\vecx}\wedge\vecz)F(\overline{\vecx}\vee\vecz)\geq F(\overline{\vecx})F(\vecz)\\
F(\mathbf 0)F(\vecz)\geq F(\vecx\wedge\vecz)F(\overline{\vecx}\wedge\vecz)\\
F(\vecz)F(\mathbf 1)\geq F(\vecx\vee\vecz)F(\overline{\vecx}\vee\vecz)
\end{align*}
In each case we have used the fact that the tuples on the
right-hand-side are not complements, or, equivalently, the tuples on
the left-hand-side are not $\mathbf 0$ and $\mathbf 1$.

Multiplying these four inequalities we get $F(\mathbf 0)F(\mathbf 1)C\geq
F(\vecx)F(\overline{\vecx})C$ where
\[C=F(\vecz)^2F(\vecx\wedge\vecz)F(\vecx\vee\vecz)F(\overline{\vecx}\wedge\vecz)F(\overline{\vecx}\vee\vecz) \]
The inequalities also imply that $C\geq
F(\vecx)F(\overline{\vecx})F(\vecz)^4$. But $F(\mathbf 0)F(\mathbf 1)<
F(\vecx)F(\overline{\vecx})$ so $C=0$ and hence $F(\vecz)=0$.
\end{proof}

\begin{lemma}\label{lem:minnonjoin}
Let $R$ be a pinning-minimal relation subject to not being closed
under joins (so there exists $\vecx,\vecy\in R$ such that
$\vecx\vee\vecy\notin R$).  Then
$R=\{\veczero,\vecx,\overline{\vecx}\}$ or
$R=\{\vecx,\overline{\vecx}\}$.
\end{lemma}
\begin{proof}
For all $\vecx,\vecy\in R$ with $\vecx\vee\vecy\notin R$, the pinning
of $R$ by $\{i\mapsto x_i\mid x_i=y_i\}$ is not closed under joins so
$\vecy=\overline{\vecx}$.  Hence there exists $\vecx$ with
$\vecx,\overline{\vecx}\in R$, and $\vecone\notin R$.  Also, taking
contrapositives, if $\vecy,\vecz\in R$ and $\vecy\neq\overline{\vecz}$
then $\vecy\vee\vecz\in R$.

Let $\vecy\in R\setminus\{\vecx,\overline{\vecx}\}$.  By the previous
paragraph, $\vecx\vee\vecy\in R$ and $\overline{\vecx}\vee\vecy\in R$.
But $(\vecx\vee\vecy)\vee(\overline{\vecx}\vee\vecy)=\vecone\notin R$,
so $\vecx\vee\vecy$ is the complement of $\overline{\vecx}\vee\vecy$.
Hence $\max(x_i,y_i)=1-\max(1-x_i,y_i)=\min(x_i,1-y_i)$ for all
variables $i$, which implies $\vecy=\veczero$.
\end{proof}

\begin{lemma}\label{lem:minnondeg}
Let $F$ be a pinning-minimal non-degenerate signature.  Then $F$ has
arity 2 or $\supp(F)=\{\vecx,\overline{\vecx}\}$ for some $\vecx$.
\end{lemma}
\begin{proof}
Pick some variable $v$ in the variable set of $F$.  Let $F_0$ and
$F_1$ be the pinnings of $F$ by $\{v\mapsto 0\}$ and $\{v\mapsto 1\}$
respectively.

For any degenerate signature $G$, the pinnings $G_0$ and $G_1$ defined
in the same way are linearly dependent.  So for all partial
configurations $\vecp$ such that $\pin {(F_0)} \vecp$ and $\pin
{(F_1)} \vecp$ are linearly independent, $\pin F \vecp$ is
non-degenerate and hence $\dom(\vecp)=\emptyset$.  Furthermore if
$\lambda F_0=\mu F_1$ for some $\lambda,\mu$ not both zero, then $F$ is
degenerate: by symmetry and scaling we can assume $\mu=1$, so
$F_1=\lambda F_0$, and $F$ is the tensor product of $F_0$ and the
arity 1 signature $U$ defined by $U(0)=1$ and $U(1)=\lambda$.  Hence
$F_0$ and $F_1$ form a pinning-minimal linearly independent signature
pair.  By Lemma \ref{lem:minlindep},
$\supp(F_0)\cup\supp(F_1)=\{\vecx,\overline{\vecx}\}$ for some
$\vecx\in\supp(F_0)$.

If $\supp(F_0)$ and $\supp(F_1)$ are $\{\vecx\}$ and
$\{\overline{\vecx}\}$ respectively (or vice versa) then
$\supp(F)=\{(0,\vecx),\overline{(0,\vecx)}\}$, so we are done.
Otherwise $\supp(F_0)$ or $\supp(F_1)$ is $\{\vecx,\overline{\vecx}\}$
which is therefore degenerate.  But a degenerate relation is
equivalent to $\{0,1\}^a\times\{0\}^b\times\{1\}^c$ for some
$a,b,c\geq 0$. Taking cardinalities we have $2=2^a$ so $a=1$. The
powers $b$ and $c$ must be zero because $x_u\neq \overline{x}_u$ for
all variables $u$. Hence $F$ has arity 2.
\end{proof}

\begin{lemma}\label{lem:dmterr}
A relation is a delta matroid if and only if its signature is
terraced.  (Recall that a relation $R$ is a \emph{delta matroid} if
for all $\vecx,\vecy\in R$ and for all $i\in \vecx\setdiff\vecy$ there
exists $j\in \vecx\setdiff\vecy$, not necessarily distinct from $i$,
such that $\vecx^{\{i,j\}}\in R$.  A signature $F$ is \emph{terraced}
if for all partial configurations $\vecp$ of $V$ and all $i,j$ in the
domain of $\vecp$, if $\pin F \vecp$ is identically zero then $\pin F
{\vecp^{\{i\}}}$ and $\pin F {\vecp^{\{j\}}}$ are linearly dependent.)
\end{lemma}
\begin{proof}
Let $R$ be a delta matroid.  Let $\vecp$ be a partial configuration
such that $\pin R \vecp$ is empty and let $i,j$ be variables on which
$\vecp$ is defined and such that $\pin R {\vecp^{\{i\}}},\pin R
{\vecp^{\{j\}}}$ are non-empty.  We will show that $\pin R
{\vecp^{\{i\}}}=\pin R {\vecp^{\{j\}}}$.  By symmetry it suffices to
show that for all $\vecx\in \pin R {\vecp^{\{i\}}}$ we have $\vecx \in
\pin R {\vecp^{\{j\}}}$.  Pick $\vecy \in \pin R {\vecp^{\{j\}}}$. By
the delta matroid property applied to
$((\vecx,\vecp^{\{i\}}),(\vecy,\vecp^{\{j\}}),i)$ there exists $d$,
such that $x_d\neq y_d$ or $d\in\{i,j\}$, and such that
$(\vecx,\vecp^{\{i\}})^{\{i,d\}}$ is in $R$. Since $\pin R \vecp$ is
empty we have $d=j$ and hence $\vecx\in \pin R {\vecp^{\{j\}}}$.

Conversely let $R$ be a relation whose signature is terraced. For all
$\vecx,\vecy\in R$ and all $d\in\vecx\setdiff\vecy$ we wish to show
that $\vecx^{\{d,d'\}}\in R$ for some $d'\in\vecx\setdiff\vecy$.  Let
$\vecy'\in R$ satisfy
$\{d\}\subseteq\vecx\setdiff\vecy'\subseteq\vecx\setdiff\vecy$ with
$|\vecx\setdiff\vecy'|$ minimal.  If $\vecx\setdiff\vecy'=\{d\}$ we
can take $d'=d$. Otherwise pick
$d'\in(\vecx\setdiff\vecy')\setminus\{d\}$.  Let $\vecp$ be the
restriction of $\vecx^{\{d\}}$ to $\{d,d'\}\cup\{i\mid x_i=y_i\}$.
Configurations $\vecz\in\pin R \vecp$ satisfy
$\{d\}\subseteq\vecx\setdiff(\vecp,\vecz)\subseteq(\vecx\setdiff\vecy')\setminus\{d'\}$,
but $|\vecx\setdiff(\vecp,\vecz)|<|\vecx\setdiff\vecy'|$ contradicts
the choice of $\vecy'$; therefore $\pin R \vecp$ is empty.  And $\pin
R {\vecp^{\{d\}}}$ and $\pin R {\vecp^{\{d'\}}}$ contain the
restrictions of $\vecx$ and $\vecy$ respectively (to
$(\vecx\setdiff\vecy')\setminus\{d,d'\}$).  Since $R$ has a terraced
signature, $\pin R {\vecp^{\{d\}}}=\pin R {\vecp^{\{d'\}}}$ so
$\vecx^{\{d,d'\}}\in R$.
\end{proof}

\begin{lemma}\label{lem:minnonimterr}
For every pinning-minimal non-IM-terraced signature, there is an
equivalent signature $F\sigf k$ such that:
\begin{itemize}
\item the pinning of $F$ by the partial configuration $\vecp$ defined
  by $\vecp(1)=0$ and $\vecp(2)=1$ is identically zero, and
\item there exists a configuration $\vecz$ of $\{3,4,\cdots,k\}$
  and a non-degenerate signature $T\sigf 2$ such that for all
  $x,y_3,\cdots,y_k\in\{0,1\}$ we have
\[F(x,x,y_3,\cdots,y_k)=\begin{cases}
                   T(y_3,x)&\text{if $\vecy=\vecz$ or $\vecy=\overline{\vecz}$}\\
                   0&\text{otherwise}
               \end{cases}\]
\end{itemize}
\end{lemma}
\begin{proof}
Consider an arbitrary pinning-minimal non-IM-terraced signature $F$.
Since $F$ is not IM-terraced there exist $\vecp,i,j$ such that $p_i=0$
and $p_j=1$ and $\pin F \vecp$ is identically zero, but $\pin F
{\vecp^{\{i\}}}$ and $\pin F {\vecp^{\{j\}}}$ are linearly
independent.  By renaming variables we can assume $i=1$ and $j=2$ and
$V(F)=\{1,2,\cdots,k\}$ for some $k$.

We will write $00$ and $11$ for the partial configurations $\{1\mapsto
0, 2\mapsto 0\}$ and $\{1\mapsto 1, 2\mapsto 1\}$ respectively, so
$\{\vecp^{\{1\}},\vecp^{\{2\}}\}=\{00,11\}$.  Let $\vecp'$ be the
restriction of $\vecp$ to $\dom(\vecp)\setminus\{1,2\}$.  Then
$F'=\pin F {\vecp'}$ is also not IM-terraced: $\pin {F'} {00}=\pin F
{\vecp^{\{i\}}}$ and $\pin {F'} {11}=\pin F {\vecp^{\{j\}}}$ are
linearly independents. Hence $\dom(\vecp)=\{1,2\}$ by minimality of
$F$.

We will argue that $(\pin F {00},\pin F {11})$ is a pinning-minimal
linearly independent signature pair.  We need to check that for any
non-empty partial configuration $\vecy$ of $\{3,\cdots,k\}$ the
pinnings $(\pin F {00})_{\vecy}$ and $(\pin F {11})_{\vecy}$ are
linearly dependent.  But $F_{\vecy}$ is IM-terraced by minimality of
$F$, and $(F_{\vecy})_{\vecp}$ is identically zero, and
$\vecp^{\{1\}}=00$ and $\vecp^{\{2\}}=11$, so $(F_{\vecy})_{00}$ and
$(F_{\vecy})_{11}$ are indeed linearly dependent because the order in
which pinnings are applied does not matter.

By Lemma \ref{lem:minlindep},
  $\supp(\pin F {00})\cup\supp(\pin F {11})=\{\vecz,\overline{\vecz}\}$ for some
  configuration $\vecz$ of $\{3,4,\cdots,k\}$.  Without loss of generality
  we may take $z_3=0$.  Set
$T(0,0)=\pin F {00}(\vecz)$  and
$T(1,0)=\pin F {00}(\overline{\vecz})$ and
$T(0,1)=\pin F {11}(\vecz)$ and
$T(1,1)=\pin F {11}(\overline{\vecz})$.
This $T$ satisfies the required expression for $F$.
Furthermore the signatures $F_{00}$ and $F_{11}$ are linearly independent,
hence so are the vectors $(T(0,0),T(1,0))$ and $(T(0,1),T(1,1))$,
and hence $T$ is non-degenerate.
\end{proof}

\begin{lemma}\label{lem:minnonterr}
For every pinning-minimal non-terraced signature there is
an equivalent signature $F\sigf k$ and a configuration
$\vecp:\{1,2\}\to\{0,1\}$ such that:
\begin{itemize}
\item $\pin F \vecp$ is identically zero 
\item there exists a configuration $\vecz:\{3,\cdots,k\}\to\{0,1\}$
  and a non-degenerate signature $T\sigf 2$ such that if $\vecx$ is one
  of the flips $\vecp^{\{1\}}$ or $\vecp^{\{2\}}$ (either
  $(1-p_1,p_2)$ or $(p_1,1-p_2)$ as elements of $\{0,1\}^{\{1,2\}}$),
  then for all $y_3,\cdots,y_k\in\{0,1\}$ we have
\[F(x_1,x_2,y_3,\cdots,y_k)=\begin{cases}
                   T(y_3,x_1)&\text{if $\vecy=\vecz$ or $\vecy=\overline{\vecz}$}\\
                   0&\text{otherwise}
               \end{cases}\]
\end{itemize}
\end{lemma}
\begin{proof}
Given a pinning-minimal non-terraced signature $G$, there exists
$\vecp,i,j$ such that $\pin G {\vecp}$ is identically zero, but $\pin
G {\vecp^{\{i\}}}$ and $\pin G {\vecp^{\{j\}}}$ are linearly
independent.  Let $S$ be the set containing: $i$ if $p_i=1$, and $j$
if $p_j=0$.  Then the flip $G^S$ is not IM-terraced: let
$\vecq=\vecp^S$; then $\pin {G^S} {\vecq}$ is identically zero but
$\pin {G^S} {\vecq^{\{i\}}}=\pin G {\vecp^{\{i\}}}$ and $\pin {G^S}
{\vecq^{\{j\}}}=\pin G {\vecp^{\{j\}}}$ are linearly independent.

Since $G$ is pinning-minimal non-terraced, every proper pinning of $G$
is terraced. Terracedness is preserved by flips, so every proper
pinning of $G^S$ is terraced, and hence IM-terraced.  So $G^S$ is
pinning-minimal non-IM-terraced.  The expression for a signature $F$
equivalent to $G^S$ is given by applying an arbitrary flip to the
expression given by Lemma \ref{lem:minnonimterr}.
\end{proof}

\begin{lemma}\label{lem:min3delta}
Let $R$ be a delta matroid that is pinning-minimal subject to not
being basically binary.  There is an $h$-maximisation of $R$
equivalent to a flip of $\PM_3$.
\end{lemma}
\begin{proof}
First note that any relation $R$ with an $h$-maximisation of $R$
equivalent to a flip of $\PM_3$ is not basically binary:
$h$-maximisation cannot make decomposable relations indecomposable,
and an indecomposable arity 3 relation cannot be basically binary.
Also, $R$ is indecomposable: if $R=R_1\times R_2$ then since $R$ is
not basically binary, either $R_1$ or $R_2$ is not basically binary,
but $R_1$ and $R_2$ are pinnings of $R$.

We will in fact show that $R$ has the ``sphere property'' that there
exists $\vecx\in\{0,1\}^3$ such that $\vecx\in\{0,1\}$ and $d=1,2$
such that $\vecx^U\notin R$ for subsets $U$ of $\{1,2,3\}$ with
$|U|<d$ and $\vecx^U\in R$ for $|U|=d$.  Then let $h(1)=2x_1-1$ and
$h(2)=2x_2-1$ and $h(3)=2x_3-1$. Observe that $S=R_{h-\max}$ consists
precisely of the three configurations $\vecx^U$ with $|U|=d$.  In
other words $S$ is a flip $\PM_3^{U'}$, where $U'=U$ if $d=1$ and
$U'=\{1,2,3\}\setminus U$ if $d=2$.

There exists a configuration not in $R$ (otherwise $R$ would be
basically binary). So there is an arity zero pinning of $R$ that is
the empty relation.  Let $\pin R \vecp$ be a maximal pinning subject
to $\pin R \vecp =\emptyset$.  For each $v\in \dom(\vecp)$ let
$\vecp'$ be the restriction $\vecp'$ of $\vecp$ to
$\dom(\vecp)\setminus\{v\}$; the pinning $\pin R {\vecp'}$ is
non-empty by maximality of $\pin R \vecp$, and hence the relations
$\pin R \vecp^{\{v\}}$ are non-empty.  The signature of $R$ is terraced by
Lemma \ref{lem:dmterr}, so $\pin R {\vecp^{\{v\}}}=\pin
R{\vecp^{\{v'\}}}$ for any $v,v'\in\dom(\vecp)$.

Recall that $R$ is indecomposable.  But if $\vecp$ has variable set
$\{v\}$ for some $v$ then $R$ is the product of $\{\vecp^{\{v\}}\}$
with $\pin R {\vecp^{\{v\}}}$.  So $\vecp$ has arity at least $2$.

If $\vecp$ has arity at least 3, split $\vecp$ as $(\vecq,\vecp')$
where $\dom(\vecq)=3$.  Pick $\vecy\in \pin R {\vecp^{\{v\}}}$ (for
any $v$).  Let $R'$ be the pinning of $R$ by both $\vecy$ and
$\vecp'$.  Note that $\vecq\notin R'$ but $\vecq^{\{v\}}\in R'$ for
all variables $v\in\dom(\vecq)$.  Hence $R=R'$ and $R'$ has the
sphere property with $d=1$.

The remaining case is that $\vecp$ has variable set $\{i,j\}$ for some
distinct variables $i,j$.  Since $R$ is indecomposable,
$R$ is not the product of $\pin R {\vecp^{\{i\}}}$ with an arity 2
relation on $\{i,j\}$.  Hence $\pin R {\vecp^{\{i\}}}$ and
$\pin R {\vecp^{\{i,j\}}}$ are linearly independent.  Let $R'$ be a
minimal pinning of $R$ such that $G=\pin {R'} {\vecp^{\{i\}}}$ and
$H=\pin {R'} {\vecp^{\{i,j\}}}$ are linearly independent.

By Lemma \ref{lem:minlindep} we have $G\cup
H=\{\vecy,\overline{\vecy}\}$ for some $\vecy$, and without loss of
generality, either $G=\{\vecy\}$ or $H=\{\overline{\vecy}\}$.
Also, to recap: $\pin R {\vecp}=\emptyset$ and $G=\pin R
{\vecp^{\{i\}}}=\pin R {\vecp^{\{j\}}}\neq\pin R {\vecp^{\{i,j\}}}=H$.
If $G=\{\vecy\}$ then by the delta matroid property applied to
$(\vecp^{\{i,j\}},\overline{\vecy})$, $(\vecp^{\{i\}},\vecy)$ and $j$
there exists $k\in\{j\}\cup\dom(\vecy)$ such that
$(\vecp^{\{i,j\}},\overline{\vecy})^{\{j,k\}}\in R$, but then $k$ must
lie in $\dom(\vecy)$ and $\overline{\vecy}^{\{k\}}\in \pin R
{\vecp^{\{j\}}}=G$.  Hence $\vecy$ has arity 1, $R$ has arity 3, and
the sphere property holds with
$\vecx=(\vecp^{\{i,j\}},\overline{\vecy})$ and $d=2$.  If
$H=\{\overline{\vecy}\}$ then by the delta matroid property applied to
$(\vecp^{\{i\}},\vecy)$, $(\vecp^{\{i\}},\overline{\vecy})$ and $j$,
there exists $k\in\{j\}\cup\dom(\vecy)$ such that
$(\vecp^{\{i\}},\vecy)^{\{j,k\}}\in R$, but then $k$ must lie in
$\dom(\vecy)$ and $\vecy^{\{k\}}\in \pin R {\vecp^{\{i,j\}}}=H$. Hence
$\vecy$ has arity 1, $R$ has arity 3, and the sphere property holds
with $\vecx=(\vecp^{\{i,j\}},\overline{\vecy})$ and $d=1$.
\end{proof}

%% file: basicthm.tex
\section{Main theorem}

\begin{reptheorem}{thm:basic}
Let $\Gamma$ be a finite set of relations.  If $\Gamma\subseteq\Neqconj$ or
every relation in $\Gamma$ is basically binary then $\cwt(\Gamma)$ is in FP.  Otherwise,
\begin{itemize}
\item If $\Gamma\subseteq\Imconj$ then $\BIS\APeq\cwt(\Gamma)$.
\item If $\Gamma\not\subseteq\Imconj$ then $\PMP\APred\cwt(\Gamma)$. If furthermore $\Gamma$ is not a set of delta matroids then $\SAT\APeq\cwt(\Gamma)$.
\end{itemize}
\end{reptheorem}
\begin{proof}
The inclusion in FP is given by Lemma \ref{lem:tract}.  We will therefore assume
that $\Gamma$ contains a relation that is not in NEQ-conj and a
relation that is not basically binary.  We will consider the four
cases depending on whether $\Gamma\subseteq\Imconj$ and whether
$\Gamma$ consists entirely of delta matroids:

\begin{tabular}{|l|l|l|}
\hline
$\Imconj$ &  delta & \\
          &  matroids & \\
\hline
yes & yes & impossible by Lemma \ref{lem:imdeltabb}\\
no & yes & $\PMP\APred\cwt(\Gamma)$ by Lemma \ref{lem:pmpterr}\\
yes & no & $\BIS\APeq\cwt(\Gamma)$ by Lemmas \ref{lem:nondelta}, \ref{lem:unbounded} and \ref{lem:biseasy}\\
no & no & $\SAT\APeq\cwt(\Gamma)$ by Lemmas \ref{lem:nondelta}, \ref{lem:unbounded} and \ref{lem:sateasy}\\
\hline
\end{tabular}
\end{proof}

\begin{lemma}\label{lem:imdeltabb}
Let $R$ be a delta matroid in IM-conj. Then $R$ is basically binary.
\end{lemma}
\begin{proof}
We may assume that $R$ is indecomposable.  Assume for contradiction
that $R$ has arity at least three.

Let $V$ be the variable set of $R$.  Note that no variables are
pinned: if there exists $i\in V$ and $c\in\{0,1\}$ such that $x_i=c$
for all $\vecx\in R$, then $R$ is the product of $\{c\}$ with the
pinning of $R$ by $\{i\mapsto c\}$, but this contradicts the
assumption that $R$ is indecomposable. Since $R$ is in IM-conj and no
variables are pinned, $R$ is a conjunction of implications of
variables. Therefore there is a subset $P$ of $V\times V$ such that
\[ R=\{\vecx \mid \text{$x_i\leq x_j$ for all $(i,j)\in P$}\} \]

Consider the undirected graph $G$ on $V$ where $i$ and $j$ are
adjacent if and only if $(i,j)$ or $(j,i)$ is in $P$.  Then $G$ has at
least three vertices, and since $R$ is indecomposable, $G$ is
connected.  Hence there is a vertex $i$ of degree at least two.  There
exist distinct variables $j,k\in V$ such that
$(i,j),(i,k)\in P$, or
$(j,i),(k,i)\in P$, or
$(j,i),(i,k)\in P$.
In the first case, there is no $\ell\in V$ such that
$\veczero^{\{i,\ell\}}\in R$.  In the second case, there is no
$\ell\in V$ such that $\vecone^{\{i,\ell\}}\in R$.  In the third case,
there is no $\ell\in V$ such that $\veczero^{\{j,\ell\}}\in R$.  But
the all-zero configuration $\veczero$ and the all-one configuration
$\vecone$ are both in $R$.  Hence the delta matroid property fails for
$R$.
\end{proof}

\begin{lemma}\label{lem:nondelta}
Let $\Gamma$ be a finite set of relations which are not all delta
matroids.  Then \[\nCSPs(\Gamma)\APred\cwt(\Gamma)\]

\end{lemma}
\begin{proof}
Let $R_1$ be a minimal non-terraced pinning of a relation in $\Gamma$.
By Lemma \ref{lem:dmterr}, $R_1$ is pinning-minimal non-terraced, so
by Lemma \ref{lem:minnonterr}, possibly after renaming variables,
there exist $p_1,p_2,z_3,\cdots,z_k\in\{0,1\}$ and a non-degenerate
signature $T \sigf 2$ such that for all
$(x_1,x_2)\in\{(1-p_1,p_2),(p_1,1-p_2)\}$ and all
$\vecy=(y_3,\cdots,y_k)\in\{0,1\}^{\{3,\cdots,k\}}$ we have
\[ R_1(x_1,x_2,\vecy)=T(y_3,x_1)1_{\{\vecz,\overline{\vecz}\}}(\vecy) \]
Define
\[ R_2(x_1,x_2,\vecy)=1_{\{(1-p_1,p_2),(p_1,1-p_2)\}}(x_1,x_2)T(y_3,x_1)1_{\{\vecz,\overline{\vecz}\}}(\vecy) \]
for all $x_1,x_2,y_3,\cdots,y_k\in\{0,1\}$.  In other words, $R_1$ and
$R_2$ agree except that the entries $R_2(1-p_1,1-p_2,\vecy)$ are zero.
Hence $R_2=(R_1)_{h-\max}$ where $h(1)=2p_1-1$ and $h(2)=2p_2-1$ and
$h(3)=\cdots=h(k)=0$.

Define
\begin{align*}
R_3(x_1,x_2,y_3)=\sum_{y_4,\cdots,y_k}R_2(x_1,x_2,\vecy)
                   =1_{\{(1-p_1,p_2),(p_1,1-p_2)\}}(x_1,x_2)T(y_3,x_1)
\end{align*}

$T$ is necessarily $0,1$-valued, but $T$ is also
non-degenerate. Hence $T$ has some zero, say $T(c,d)=0$. Define
\[R_4(x_1,x_2,y_3)=1_{\{(1-p_1,p_2),(p_1,1-p_2)\}}(x_1,x_2)1_{\{(1-c,d),(c,1-d)\}}(y_3,x_1)\]
So $R_3$ and $R_4$ agree except that $R_4(1-d,x_2,1-c)$ is zero for
$x_2=0,1$. Hence $R_4=(R_3)_{h-\max}$ where $h(1)=2c-1$ and
$h(2)=2d-1$ and $h(3)=0$.  By Lemma \ref{lem:hamweight} and Lemma
\ref{lem:express} we have
\[ \cwt(\Gamma\cup\{R_4\})\APred\cwt(\Gamma) \]
Crucially $R_4$ is a conjunction of an equality or disequality on the
first two variables, with an equality or disequality on the last two
variables. This implies that $R_4$ consists of two complementary
configurations and so $R_4$ is equivalent to $\{(0,0,0),(1,1,1)\}$ or
$\{(0,1,1),(1,0,0)\}$.

The rest of the proof is what is called ``2-simulating equality'' in \cite{bdddeg}.
We are given an instance of $\nCSPs(\Gamma)$, which can be written as
\[\zwp{w}{\phi}=\sum_{\vecx:V\to\{0,1\}}\left(\prod_{v\in V} w(v)_{x_v} \right)\left(\prod_{i\in I}F_i((x_{\scope(i,j)})_{j\in V(F_i)})\right) \]
We can assume that every variable has degree at least one. (Otherwise
let $\hat V=\{v\in V\mid \deg(v)>0\}$ and let $\hat w$ be the
restriction of $w$ to $\hat V$; then $\zwp{w}{\phi}=\zwp{\hat
  w}{\phi}\prod_{v\in V\setminus \hat V}(w(v)_0+w(v)_1)$.)  Modify
$\phi$ as follows to produce a new $(=2)$-formula $\phi'$ on a
variable set $V'$.  For each variable $v\in V$, replace the
$d=\deg(v)$ uses of $v$ by separate variables $v_1,\cdots,v_d$ and
insert new atomic formulas $R_4(v_i,v_i,u_{i+1}))$ for
$i=1,\cdots,d$, where $u_{d+1}=u_1$, to obtain a new formula $\phi'$.
Note that every variable in $\phi'$ is used exactly twice.  Set
$w'(v_1)=w(v)$ for all $v\in V$, and
$w'(v_2)=\cdots=w'(v_d)=w(u_1)=\cdots=w(u_d)=1$.  Then
$\zwp{w}{\phi}=\zwp{w'}{\phi'}$: the contributions to
$\zwp{w'}{\phi'}$ come from configurations where for each $v$ the
variables $v_1,\cdots,v_d$ get the same value $x_v$, and these
configurations have the same weight as the corresponding configuration
$\vecx$ in $\zwp{w}{\phi}$.  And we can just call the
$\cwt(\Gamma\cup\{R_4\})$ oracle to obtain $\zwp{w'}{\phi'}$.
\end{proof}

\begin{lemma}\cite[Proposition 25]{lsm}\label{lem:biseasy}
Let $\calF$ be a finite subset of IM-conj. Then $\nCSPs(\calF)\leq \BIS$.
\end{lemma}
\begin{proof}
We have $\nCSP(\{\IMP\})\APred\BIS$ by \cite[Theorem 3]{csptrich}, so
it suffices to show that $\nCSPs(\{\IMP\})\APred\nCSP(\{\IMP\})$.  The
construction in \cite[Proposition 25]{lsm} simulates an arbitrary
polynomial-time computable arity 1 signature using $\IMP$, in
polynomial time.
\end{proof}

In \cite{fisher:1776} it is shown that the problem of counting perfect
matchings reduces to counting perfect matching of graphs of maximum
degree three. Hence:
\begin{lemma}\cite{fisher:1776}\label{lem:pmptopm3}
$\PMP\leq\nCSP_{=2}(\{PM_3\})$.
\end{lemma}

\begin{lemma}\label{lem:pmpterr}
Let $R$ be a delta matroid that is not basically binary.  Then
$\PMP\APred\nCSPs_{=2}(\{R\})$.
\end{lemma}
\begin{proof}
By Lemma \ref{lem:min3delta} and Lemma \ref{lem:freepin} we can assume
that $R$ has arity 3 and there exists $\vecx\in\{0,1\}$ and $d=1,2$
such that $\vecx^U\notin R$ for subsets $U$ of $\{1,2,3\}$ with
$|U|<d$ and $\vecx^U\in R$ for $|U|=d$.  Let $h(1)=2x_1-1$ and
$h(2)=2x_2-1$ and $h(3)=2x_3-1$. Then $S=R_{h-\max}$ consists
precisely of the three configurations $\vecx^U$ with $|U|=d$.  In
other words $S$ is a flip $\PM_3^{U'}$. By Lemma \ref{lem:hamweight} we
have
\[\nCSPs_{=2}(\{\PM_3^{U'}\})\APred \nCSPs_{=2}(\{R\})\]
If $|U'|\leq 1$ let $U''=U'$.  Otherwise let $U''=\{1,2,3\}\setminus
U'$; the complexity is not changed by exchanging the roles of $0$ and
$1$:
\[\nCSPs_{=2}(\{\PM_3^{U''}\})\APeq \nCSPs_{=2}(\{U'\})\]
In either case $|U''|\leq 1$.  If $|U''|=1$, reorder the variables if
necessary we can assume $U''=\{1\}$.  In this case $\PM_3$ can be
expressed by a 2-formula over $\{\PM_3^{\{1\}},\PIN_1\}$:
\[ \NEQ(y,z)=\sum_x \PM_3^{\{1\}}(x,y,z)\PIN_1(x) \]
\[\PM_3(x,y,z)=\sum_{x'}\NEQ(x,x')\PM_3^{\{1\}}(x',y,z) \]
Hence by Lemma \ref{lem:freepin} and Lemma \ref{lem:express} we have
$\cwt(\{\PM_3\})\APred\cwt(\{\PM_3^{\{1\}}\})$.  In any case it
suffices to show that $\PMP\APred\cwt(\{\PM_3\})$, which is Lemma
\ref{lem:pmptopm3}.
\end{proof}

%% file: weighted.tex
\section{An extension to signatures}\label{sec:weight}

In this section we will give the extensions of Theorem \ref{thm:basic}
mentioned in the introduction.

This section is quite technical, so here is a quick summary.  We work
in the setting of finite sets of variable weights as much as possible.
We then collect all our results for arbitrary variable weights in
Theorem \ref{thm:weightthm}, and collect all our results for finite
sets of variable weights in Theorem \ref{thm:finitewts}.  First of
all, Lemma \ref{lem:holanty} uses certain non-IM-terraced signatures
to reduce a slightly different unbounded-degree problem
``$\nCSP(T^{\otimes}\calF^B)$'' defined below, to a degree-two problem
$\nCSP_{=2}(\calF)$, using an adaptation of the Holant theorem as used
in \cite{holantc}.  Lemma \ref{lem:insertwt} provides unary signatures
in this unbounded-degree problem.  Lemma \ref{lem:nonterruse} ties the
previous two Lemmas together and extends to any non-IM-terraced
signature. Lemma \ref{lem:bissatwt} applies this to reducing $\BIS$
and $\SAT$ to certain $\nCSP$s.  For infinite sets of variables
weights, Lemma \ref{lem:pmpterrwt} reduces $\PMP$ to certain $\nCSP$s,
and Lemma \ref{lem:neqing} uses $h$-maximization to provide flips in
some cases, which means non-terraced signatures are as useful as
non-IM-terraced signatures in that setting.

Let $T\sigf 2$ and let $F$ be a signature. The following construction
is used for holographic transformations of Holant problems (see for
example \cite{holantc}), and is usually denoted $T^{\otimes k}F$ if
$F\sigf k$.  But it will be convenient not to include the arities $k$.
Define $T^{\otimes}F\sigf V$ by
\[
(T^{\otimes}F)(\vecx)=\sum_{\vecy}\left(\prod_{i\in V(F)} T(x_i,y_i)\right)F(\vecy)
\]

Let $B=1$ or $B=2$ and let $T$ be a non-degenerate arity 2 signature.
(To make the results stronger we will work with $\nCSP_{=2}$ (or
Holant) problems rather than $\nCSP_{\leq 2}$.  This is indirectly why
we end up using the technical complication of the $B=2$ case.)  In
this section we will use the notation $T^{\otimes}F^B$, where $F$ is a
signature, to denote $T^{\otimes}F'$ where $F'(\vecx)=F(\vecx)^B$.  We
will use the notation $T^{\otimes}\calF^B$, where $\calF$ is a set of
signatures, to denote $\{T^{\otimes}F^B\mid F\in\calF\}$.

\begin{lemma}\label{lem:holanty}
Let $B=1$ or $B=2$.  Let $T:\{0,1\}^2\to\Rnonneg$.  Let
$G:\{0,1\}^{B+2}\to\Rnonneg$ be a signature such that for all
$x,y_1,\cdots,y_B\in\{0,1\}$ we have $G(1,0,y_1,\cdots,y_B)=0$ and
\[ G(x,x,y_1,\cdots,y_B)=\EQ_B(y_1,\cdots,y_B)T(y_1,x) \]
Then \[ \nCSP(T^{\otimes}\calF^B)\APred \nCSP_{=2}(\calF)\]
\end{lemma}
\begin{proof}
Let $\phi$ be an instance of $\nCSP(T^{\otimes}\calF^B)$.  We may
assume that every variable has non-zero degree. 

\def\use{\operatorname{use}}

We will enumerate each use of each variable in the following way.  Let
$V=V^\phi$, $I=I^\phi$, $F=F^\phi$ and $\scope=\scope^\phi$.
Define $L=\{(v,d)\mid v\in V, 1\leq d\leq \deg^\phi(v)\}$ and
$R=\{(i,j)\mid i\in I, j\in V(F_i)\}$.  There is a bijection
$\use:L\to R$ such that $\scope(\use(v,d))=v$ for all $(v,d)\in
L$.  We wish to compute $Z=\zp{\phi}$, which is
\[ \sum_{\vecz\in\{0,1\}^V}\prod_{i\in I} (TF_i)((z^{\scope(i,j)})_{j\in V(F_i)})\]

For the rest of the proof, product indices $i,j,b,v,d$ will range over
$i\in I$ and $j\in V(F_i)$ and $1\leq b\leq B$ and $v\in V$ and $1\leq
d\leq \deg(v)$.  The variables $\vecx,\vecy,\vecz$ range over
$\vecx:L\to\{0,1\}$ and $\vecy:R\times\{1,\cdots,B\}\to\{0,1\}$ and
$\vecz:V\to\{0,1\}$, and $y^{(i,j)}_b$ means $y_{(i,j),b}$, and
$x^{(v,0)}$ means $x^{(v,\deg(v))}$.  Define $\phi'$ to be the
$(=2)$-formula given by
\[\zp()=\sum_{\vecx,\vecy} \left(\prod_{v,d} G(x^{(v,d-1)},x^{(v,d)},y^{\use(v,d)}_1,\cdots,y^{\use(v,d)}_B)\right)
   \left(\prod_{i,b} F_i((y^{(i,j)}_b)_{j\in V(F_i)})\right) \]

The reduction queries the $\nCSP_{=2}(\calF)$ oracle on $\phi'$,
passing through the error parameter, and returns the result.  To show
that the reduction is correct we must show that
$\zp{\phi}=\zp{\phi'}$. This is mostly algebraic manipulation with the
products below.

\begin{align*}
\mathrm{ZTerms}(\vecz)&=\prod_i (T^{\otimes}F_i^B)((z^{\scope(i,j)})_{j\in V(F_i)})\\
\mathrm{YZTrans}(\vecy,\vecz)&=\prod_{i,j} \EQ_B(y^{(i,j)}_1,\cdots,y^{(i,j)}_B)T(y^{(i,j)}_1,z^{\scope(i,j)})\\
\mathrm{YTerms}(\vecy)&=\prod_{i,b} F_i((y^{(i,j)}_b)_{j\in V(F_i)})\\
\mathrm{XEq}(\vecx)&=\prod_v \EQ_{\deg(v)}(x^{(v,1)},\cdots,x^{(v,\deg(v))})\\
\mathrm{XYTrans}(\vecx,\vecy)&=\prod_{v,d} \EQ_B(y^{\use(v,d)}_1,\cdots,y^{\use(v,d)}_B)T(y^{\use(v,d)}_1,x^{(v,d)})\\
\mathrm{XYGTrans}(\vecx,\vecy)&=\prod_{v,d} G(x^{(v,d-1)},x^{(v,d)},y^{\use(v,d)}_1,\cdots,y^{\use(v,d)}_B)
\end{align*}

Note:

\begin{enumerate}
\item
For fixed $\vecz$ we have
$\mathrm{ZTerms}(\vecz)=\sum_{\vecy}\mathrm{YZTrans}(\vecy,\vecz)\mathrm{YTerms}(\vecy)$
by expanding the definition of $TF_i^B$.

\item
Summing over $\vecx$ with the factor $\mathrm{XEq}(\vecx)$ is the same
as summing over $\vecz$ and defining $\vecx$ by $x^{(v,d)}=z_v$.
Hence summing over $\vecx$ with the factor
$\mathrm{XEq}(\vecx)\mathrm{XYTrans}(\vecx,\vecy)$ is the same as
summing over $\vecz$ with the factor $\mathrm{YZTrans}(\vecy,\vecz)$.

\item
Fix $\vecx$ and $\vecy$.  If $\mathrm{XEq}(\vecx)=1$ then
$\mathrm{XYTrans}(\vecx,\vecy)=\mathrm{XYGTrans}(\vecx,\vecy)$ by
definition of $G$.  And if $\mathrm{XEq}(\vecx)$ is zero then so is
$\mathrm{XYGTrans}(\vecx,\vecy)$. Hence $\mathrm{XEq}(\vecx)
\mathrm{XYTrans}(\vecx,\vecy)=\mathrm{XYGTrans}(\vecx,\vecy)=0$.
\end{enumerate}

Hence

\begin{align*}
\zp{\phi}
&=\sum_{\vecz} \mathrm{ZTerms}(\vecz)\\
&=\sum_{\vecy,\vecz} \mathrm{YZTrans}(\vecy,\vecz)\mathrm{YTerms}(\vecy)\\
&=\sum_{\vecx,\vecy} \mathrm{XEq}(\vecx) \mathrm{XYTrans}(\vecx,\vecy) \mathrm{YTerms}(\vecy)\\
&=\sum_{\vecx,\vecy} \mathrm{XYGTrans}(\vecx,\vecy) \mathrm{YTerms}(\vecy)\\
&=\zp{\phi'}\qedhere
\end{align*}
\end{proof}

\begin{lemma}\label{lem:insertwt}
Let $B=1$ or $B=2$ and let $T\sigf 2$ be a non-degenerate arity 2
signature.  Assume that there exist $i,j\in\{0,1\}$ such that
$T(0,i)>T(1,i)$ and $T(0,j)<T(1,j)$.  Let $U(0),U(1)>0$ and let $F$ be
any signature with $|\supp(F)|>1$.  There exists a simple weighting
$G$ of $F$ such that $U$ is defined by a pps-formula over
$T^{\otimes}G^B$.
\end{lemma}
\begin{proof}
Let $\vecx,\vecx'$ be distinct tuples in $\supp(F)$.  By taking an
equivalent signature if necessary we can assume that
$V(F)=\{1,\cdots,n\}$ for some $n$ and that $x_n\neq x'_n$.  Let
\[ H(y_n)=\sum_{x_1,y_1,\cdots,x_{n-1},y_{n-1}} T(x_1,y_1)\cdots T(x_{n-1},y_{n-1}) F(y_1,\cdots,y_n)^B \]
Note that $H(0),H(1)>0$.

Let $\det T=T(0,0)T(1,1)-T(0,1)T(1,0)$. We will argue that there is an
integer $m>0$ and polynomial-time computable reals $W(0),W(1)>0$ such
that
\begin{align}
\begin{pmatrix}H(0)W(0)^B\\H(1)W(1)^B\end{pmatrix} =
   \frac{1}{\det T}\begin{pmatrix}T(1,1)&-T(0,1)\\-T(1,0)&T(0,0)\end{pmatrix}
                             \begin{pmatrix}U(0)^{1/m}\\U(1)^{1/m}\end{pmatrix}
\label{eqn:unaryeqn}
\end{align}

We just need to check that the right-hand-side of \eqref{eqn:unaryeqn}
has non-negative entries.  There are two cases.  If $T(1,1)>T(0,1)$
and $T(0,0)>T(1,0)$ then $\det T$ is positive and for sufficiently
large $m$ we have $T(1,1)U(0)^{1/m} > T(0,1)U(1)^{1/m}$ and
$T(1,0)U(0)^{1/m} < T(0,0)U(1)^{1/m}$.  If $T(1,1)<T(0,1)$ and
$T(0,0)<T(1,0)$ then $\det T$ is negative and for sufficiently large
$m$ we have $T(1,1)U(0)^{1/m} < T(0,1)U(1)^{1/m}$ and
$T(1,0)U(0)^{1/m} > T(0,0)U(1)^{1/m}$.  In either case the
right-hand-side of \eqref{eqn:unaryeqn} has non-negative entries.

With these $m,W(0),W(1)$ we have
\[ \begin{pmatrix}U(0)^{1/m}\\U(1)^{1/m}\end{pmatrix} =
   \begin{pmatrix}T(0,0)&T(0,1)\\T(1,0)&T(1,1)\end{pmatrix}\begin{pmatrix}H(0)W(0)^B\\H(1)W(1)^B\end{pmatrix} \]

Define $G\sigf n$ by $G(\vecx)=F(\vecx)W(x_n)$.  Then for all $x_1\in\{0,1\}$,
\[ U(x_n)=\left( \sum_{y_n} T(x_n,y_n) H(y_n) W(y_n)^B\right)^m=\left(\sum_{x_1,\cdots,x_{n-1}} (T^{\otimes}G^B)(\vecx)\right)^m  \]
By distributivity the right-hand-side can be written as a pps-formula
over $T^{\otimes}G^B$.
\end{proof}

\begin{lemma}\label{lem:nonterruse}
Let $\calF$ be a finite set of signatures containing a non-IM-terraced
signature.  There exists $B\in\{1,2\}$ and a non-degenerate arity 2 signature
$T\sigf 2$ such that for all finite sets of arity 1 signatures $S$ there
is a finite set of variable weights $W$ such that
\[\nCSP(T^\otimes\calF^B\cup S)\APred \nCSP^W_{=2}(\calF)\]
\end{lemma}
\begin{proof}
By Lemma \ref{lem:freepin} we can assume that $\calF$ is closed under
pinnings.  Choose a pinning-minimal non-IM-terraced signature
$F\in\calF$. Renaming the variable set if necessary, $F$ has the form
given by Lemma \ref{lem:minnonimterr} and in particular there exists
$T$ and $\vecz\in\{0,1\}^B$, $B\geq 1$ such that $z_1=0$ and for all
$x=0,1$ and all $\vecy\in\{0,1\}^B$ we have
\[F(x,x,\vecy)=1_{\{\vecz,\overline{\vecz}\}}(\vecy)T(y_1,x_1)\]
If $B\geq 3$ there are $1\leq i<j\leq B$ with $y_i=y_j$ and we can
express the non-IM-terraced signature $F'$ defined by:
\begin{align*}
&F'(x_1,x_2,y_1,\cdots,y_{i-1},y_{i+1},\cdots,y_{j-1},y_{j+1},\cdots,y_B)=\\
  &\sum_y F(x_1,x_2,y_1,\cdots,y_{i-1},y,y_{i+1},\cdots,y_{j-1},y,y_{j+1},\cdots,y_B)
\end{align*}
Repeating this and using Lemma \ref{lem:express} we can assume $B\leq
2$.

If $B=2$ and $z_1\neq z_2$, define $F'$ by
\[F'(x_1,x_2,y_1,y_2)=\sum_{t,y_2'} F(x_1,x_2,y_1,y_2')F(t,t,y_2',y_2)\]
Then $F'(1,0,y_1,y_2)=0$ for all $y_1,y_2\in\{0,1\}$. Also,
for all $x,y_1,y_2\in\{0,1\}$,
\[F'(x,x,y_1,1-y_2)=1_{\{(0,0),(1,1)\}}(y_1,y_2)\left(\sum_t F(t,t,1-y_1,y_1)\right)T(y_1,x)\]
By Lemma \ref{lem:express} we can use $F'$ instead of $F$.  Therefore
we can assume that $\vecz$ is either $(0)$ or $(0,0)$.

Furthermore by taking a simple weighting of $F$ and invoking Lemma
\ref{lem:simpleweight}, we can assume that there exist $i,j\in\{0,1\}$
such that $T(0,i)>T(1,i)$ and $T(0,j)<T(1,j)$.  Indeed let
$U(0)=T(1,0)+T(1,1)$ and $U(1)=T(0,0)+T(0,1)$.  Replacing $F$ by the
simple weighting $F'$ defined by
\[ F'(x_1,x_2,y_1,y_2)=U(y_1)F(x_1,x_2,y_1,y_2) \]
has the effect of replacing $T(y,z)$ by $U(y)T(y,z)$.  If
$T(0,0)T(1,1)>T(0,1)T(1,0)$ then $U(0)T(0,0) > U(1)T(1,0)$ and
$U(0)T(0,1) < U(1)T(1,1)$. Otherwise $T(0,0)T(1,1)<T(0,1)T(1,0)$ so
$U(0)T(0,0) < U(1)T(1,0)$ and $U(0)T(0,1) > U(1)T(1,1)$.

Let $S'$ be the set of permissive signatures in $S\cup\{U_0,U_1\}$
where $U_0(0)=2,U_0(1)=1$ and $U_1(0)=1,U_1(1)=2$.
For each $U\in S'$, let $F_U$ be the signature given by Lemma
\ref{lem:insertwt} such that $F_U$ is a simple weighting of $F$ (or
any other signature in $\calF$ - we only use $F$ for concreteness), and
$U$ can be expressed by a pps-formula over $\{T^{\otimes}F_U^B\}$.  Let $\calG=\calF\cup\{F_U\mid U\in S'\}$.  By
Lemma \ref{lem:express}, Lemma \ref{lem:holanty}, and Lemma
\ref{lem:simpleweight}, we have
\[ \nCSP(T^{\otimes}\calF^B\cup S')\APred\nCSP(T^{\otimes}\calG^B)
   \APred\nCSP_{=2}(\calG)\APred \nCSP^W_{=2}(\calF) \] for some
   finite set $W$.
Using $U_0$ and $U_1$ as variable weights we have:
\[ \nCSP^{\{(2^a,2^b)\mid a,b\in\mathbb{Z}\}}(T^{\otimes}\calF^B\cup S')\APred\nCSP(T^{\otimes}\calF^B\cup S') \]
But $\PIN_0=(U_1)_{h-\max}$ with $h(1,0)=1$ and $h(1,1)=0$, and
similarly $\PIN_1=(U_0)_{h-\max}$ with $h(1,0)=0$ and $h(1,1)=1$, so
by Lemma \ref{lem:hamweight} we have
\[\nCSP(T^{\otimes}\calF^B\cup S'\cup\{\PIN_0,\PIN_1\})\APred \nCSP^{\{(2^a,2^b)\mid a,b\in\mathbb{Z}\}}(T^{\otimes}\calF^B\cup S') \]
The signatures in $S\setminus S'$ are just scalar multiples of
$\PIN_0$ and $\PIN_1$ so we have established that $\nCSP(T^{\otimes}\calF^B\cup
S)\APred \nCSP^W_{=2}(\calF)$.
\end{proof}

\begin{lemma}\label{lem:bissatwt}
Let $\calF$ be a finite set of signatures. Assume that $\calF$
contains a signature that is not in Weighted-NEQ-conj and a signature
that is not IM-terraced.  Let $\Xprob=\BIS$ if every signature in $\calF$
is logsupermodular, and let $\Xprob=\SAT$ otherwise.  There is a finite
set of variable weights $W$ such that
\[\Xprob\APred\nCSP^W_{=2}(\calF)\]
\end{lemma}
\begin{proof}
By Lemma \ref{lem:freepin} we can assume $\calF$ is closed under
pinnings.  Let $G$ be a pinning-minimal signature subject to
$G\in\calF\setminus\text{Weighted-NEQ-conj}$. In particular $G$ is
indecomposable.  As in Proposition {prop:wncisclone} we will use the
characterization that an indecomposable signature is in
Weighted-NEQ-conj if and only if its support has order at most two.

Let $B,T$ be as given by Lemma \ref{lem:nonterruse} applied to
$\calF$.  Either $T(0,0)T(1,1)>0$ and $\supp(T^{\otimes}G^B)\supseteq
\supp(G)$, or $T(0,1)T(1,0)>0$ and $\supp(T^{\otimes}G^B)\supseteq
\{\overline{\vecx}\mid\vecx\in\supp(G)\}$.  In either case
$|\supp(T^{\otimes}G^B)|\geq|\supp(G)|>2$.  If
$T^{\otimes}G^B=G_1\otimes G_2$ then $G^B=(S^{\otimes}G_1)\otimes
(S^{\otimes}G_2)$ where $S$ is the matrix inverse of $T$, that is, the
unique solution to $\sum_jT(i,j)S(j,k)=\EQ_2(i,k)$ ($i,k\in\{0,1\}$).
But $G^B$ is indecomposable.  Therefore $T^{\otimes}G^B$ is
indecomposable, and hence it is not in Weighted-NEQ-conj.

If $\Xprob=\SAT$ then let $H$ be a pinning-minimal non-logsupermodular
signature in $\calF$.  In particular by Lemma \ref{lem:minnonlsm},
$\supp(H)\subseteq\{\veczero,\vecx,\overline{\vecx},\vecone\}$ for
some vector $\vecx$ with $a$ zeros and $b$ ones for some $a,b\geq
1$. Hence
\begin{align*}
&\begin{pmatrix}
(T^{\otimes}H^B)(\veczero)&
(T^{\otimes}H^B)(\vecx)\\
(T^{\otimes}H^B)(\overline{\vecx})&
(T^{\otimes}H^B)(\vecone)
\end{pmatrix}\\
&=\begin{pmatrix}
T(0,0)^a&T(0,1)^a\\
T(1,0)^a&T(1,1)^a
\end{pmatrix}
\begin{pmatrix}
H(\veczero)^B&
H(\vecx)^B\\
H(\overline{\vecx})^B&
H(\vecone)^B
\end{pmatrix}
\begin{pmatrix}
T(0,0)^b&T(1,0)^b\\
T(0,1)^b&T(1,1)^b
\end{pmatrix}
\end{align*}
Denote the latter expression by $M_1M_2M_3$.  Since
$H(\veczero)H(\vecone)<H(\vecx)H(\overline{\vecx})$, the middle matrix
$M_2$ has a negative determinant.  The determinants of the neighbouring
matrices $M_1$ and $M_3$ have the same sign: if
$T(0,0)T(1,1)>T(0,1)T(1,0)$ they both have a positive determinant,
otherwise they both have a negative determinant.  Therefore the matrix
on the left-hand-side has a negative determinant, and hence
$T^{\otimes}H^B$ is not logsupermodular.

Let $B,T$ be as given by Lemma \ref{lem:nonterruse} applied to
$\calF$.  By Lemma \ref{lem:unbounded} there is a finite set of arity 1
signatures $S$ such that $\Xprob\APred\nCSP(T^{\otimes}\calF^B\cup S)$.  By the
choice of $B$ and $T$ we have $\nCSP(T^{\otimes}\calF^B\cup
S)\APred\nCSP^W_{=2}(\calF)$ for some finite set $W$.
\end{proof}

\begin{lemma}\label{lem:pmpterrwt}
Let $F$ be a terraced signature whose support is not basically binary.
Then $\PMP\APred\nCSPs_{=2}(\{F\})$.
\end{lemma}
\begin{proof}
By pinning and applying $h$-maximisation as in the proof of Lemma
\ref{lem:pmpterr}, we can assume $\supp(F)=\PM_3^U$ for some
$U\subseteq\{1,2,3\}$.

We will show that $\PM_3^U$ is a simple weighting of $F$.  Let $F'=F^U$
and $U'_1(0)=U'_2(0)=U'_3(0)=1$ and $U'_1(1)=1/F'(1,0,0)$,
$U'_2(1)=1/F'(0,1,0)$, $U'_3(1)=1/F'(0,0,1)$.  Then
$F'(x_1,x_2,x_3)U'_1(x_1)U'_2(x_2)U'_3(x_3)=\PM_3(x_1,x_2,x_3)$ for
all $x_1,x_2,x_3\in\{0,1\}$.  For all $x\in\{0,1\}$ define
$U_i(x)=U'_i(x)$ for $i\in\{1,2,3\}\setminus U$ and $U_i(x)=U'_i(1-x)$
for $i\in U$. Then
$F(x_1,x_2,x_3)U_1(x_1)U_2(x_2)U_3(x_3)=\PM_3^U(x_1,x_2,x_3)$ for all
$x_1,x_2,x_3\in\{0,1\}$ as required.

By Lemma \ref{lem:pmpterr} there is an
AP-reduction from $\PMP$ to $\nCSP_{=2}(\{\PM_3^U\})$;
and since $\PM_3^U$ is a simple weighting of $F$, by Lemma
\ref{lem:simpleweight} there is an AP-reduction from
$\nCSP_{=2}(\{\PM_3^U\})$ to $\nCSPs_{=2}(\{F\})$.
\end{proof}

\begin{lemma}\label{lem:neqing}
Let $\calF$ be a finite set of signatures, containing a signature
whose support is not in IM-conj.  Then
$\nCSPs_{=2}(\calF')\APred\nCSPs_{=2}(\calF)$ where $\calF'$ is the
closure of $\calF$ under flips:
\[\calF'=\{F^U\mid \text{$F\in\calF$ and $U\subseteq V(F)$}\}\]
\end{lemma}
\begin{proof}
It suffices to do one flip at a time: to show that for all $G\in
\calF$ and all $U\subseteq V(G)$ we have
$\nCSP(\calF\cup\{G^U\})\APred\nCSP(\calF)$.  By Lemma
\ref{lem:freepin} and Lemma \ref{lem:hamweight} we can assume that
$\calF$ is closed under pinnings and $h$-maximisations.

Pick a pinning-minimal signature $F\in\calF$ such that $\supp(F)$ is
not IM-conj.  By Lemma \ref{lem:minnonlsm} $\supp(F)$ is a (proper)
subset of $\{\veczero,\vecz,\overline{\vecz},\vecone\}$ for some
$\vecz$, and by taking an equivalent signature we can assume that
there exists $a,b\geq 1$ such that $\vecz$ is an arity $a+b$ vector
with $z_i=0$ for $1\leq i\leq a$ and $z_i=1$ for $a+1\leq i\leq a+b$.
By taking a suitable $h$-maximisation (Lemma \ref{lem:hamweight}) we
may assume that $\veczero,\vecone\notin\supp(F)$, and by simple
weighting (Lemma \ref{lem:simpleweight}) we may assume that $F$ is
zero-one valued.

If the arity of $F$ is two then $F=\NEQ$. But
\[ G^U(\vecx,\vecx')=\sum_{\vecy:U\to\{0,1\}}G(\vecy,\vecx')\prod_{i\in U}\NEQ(x_i,y_i) \]
for all $\vecx\in\{0,1\}^U$ and $\vecx'\in\{0,1\}^{V(G)\setminus U}$.
Hence $\nCSPs_{=2}(\calF\cup\{G^U\})\APred\nCSPs_{=2}(\calF)$ by Lemma
\ref{lem:express}.

If the arity of $F$ is greater than two then for all
$\vecx,\vecx'\in\{0,1\}^a$ and $\vecy,\vecy'\in\{0,1\}^b$ we have
\[\EQ_{2a}(\vecx,\vecx')=\sum_{\vecy}F(\vecx,\vecy)F(\vecy,\vecx')\]
\[\EQ_{2b}(\vecy,\vecy')=\sum_{\vecx}F(\vecx,\vecy)F(\vecy',\vecx)\]
One of these has arity at least three, so Lemma \ref{lem:holanty} can
be applied.  The Lemma is trivial in $\calF$ is contained in
Weighted-NEQ-conj, and otherwise by Lemma \ref{lem:sateasy}, Lemma
\ref{lem:unbounded} and Lemma \ref{lem:holanty} we have
\[\nCSPs_{= 2}(\calF')\APred\SAT\APred\nCSPs(\calF)\APred\nCSPs_{=2}(\calF)\]
as required.  And if $\calF$ is contained in Weighted-NEQ-conj then
$\nCSPs_{= 2}(\calF')$ already has an FPRAS again by Lemma
\ref{lem:unbounded}.
\end{proof}

\begin{reptheorem}{thm:weightthm}
Let $\calF$ be a finite set of signatures.  If every signature in
$\calF$ is basically binary or every signature in $\calF$ is
in Weighted-NEQ-conj, then $\nCSPs_{\leq 2}(\calF)$ has an FPRAS.  Otherwise
assume furthermore that there is a signature in $\calF$ that is not
terraced or that does not have basically binary support. Then:
\begin{itemize}
\item If every signature in $\calF$ is logsupermodular then
  $\BIS\APred\nCSPs_{=2}(\calF)$ (and $\calF$ necessarily contains a
  signature that is not terraced).
\item If some signature in $\calF$ is not logsupermodular then
  $\PMp\APred\nCSPs_{=2}(\calF)$. If furthermore some
  signature in $\calF$ is not terraced then $\SAT\APeq\nCSPs_{=2}(\calF)$.
\end{itemize}
\end{reptheorem}

\begin{proof}
The FPRAS is given by Lemma \ref{lem:tract}.  If every signature in
$\calF$ is terraced, there is a signature $F$ in $\calF$ that does not
have basically binary support.  By Lemma \ref{lem:imdeltabb} $F$ is not
logsupermodular, and by Lemma \ref{lem:pmpterrwt} we have
$\PMP\APred\nCSPs_{=2}(\{F\})$ as required.

Otherwise, there is a signature $F$ in $\calF$ that is not terraced.
By definition there is a pinning $\pin F \vecp$ and there are
variables $i,j\in\dom(\vecp)$ such that $\pin F \vecp$ is identically
zero but $\pin F {\vecp^{\{i\}}}$ and $\pin F {\vecp^{\{j\}}}$ are
linearly independent.  If $p_i\neq p_j$ then $F$ is not IM-terraced.
Otherwise $p_i=p_j$.  There are
$\vecx,\vecy\in\{0,1\}^{V(F)\setminus\dom(\vecp)}$ such that
$F(\vecp^{\{i\}},\vecx)$ and $F(\vecp^{\{j\}},\vecy)$ are non-zero,
but $F(\vecp,\vecx\vee\vecy)=F(\vecp,\vecx\wedge\vecy)=0$.  Hence
$\supp(F)$ is not in IM-conj, The flip $F^{\{i\}}$ is not IM-terraced,
and $\nCSPs_{=2}(\calF\cup\{F^{\{i\}}\})\APred\nCSPs_{=2}(\calF)$ by
Lemma \ref{lem:neqing}.

So in either case we can assume that $F$ is not IM-terraced.  By
Lemma \ref{lem:bissatwt} $\BIS\APred\nCSPs_{=2}(\calF)$, and
$\SAT\APeq\nCSPs_{=2}(\calF)$ if $\calF$ contains a signature that is
not logsupermodular.
\end{proof}

\begin{reptheorem}{thm:finitewts}
Let $\calF$ be a finite set of signatures. Assume that not every
signature in $\calF$ is in Weighted-NEQ-conj, and not every signature
in $\calF$ is basically binary, and not every signature in $\calF$ is
terraced.
(This the same setting as the $\BIS$ and $\SAT$ reductions in Theorem
\ref{thm:weightthm}.)

Unless all the following conditions hold, there is a finite set
$W\subseteq\R_p\times\R_p$ such that
$\Xprob\APred\nCSP^W_{=2}(\calF)$ where $\Xprob=\BIS$ if every signature in
$\calF$ is logsupermodular, and $\Xprob=\SAT$ otherwise.

\begin{enumerate}
\item Every signature $F\in\calF$ is IM-terraced.
\item Either the support of every signature $F$ in $\calF$ is closed
  under meets
  ($\vecx,\vecy\in\supp(F)\implies\vecx\wedge\vecy\in\supp(F)$), or
  the support of every signature $F$ in $\calF$ is closed under joins
  ($\vecx,\vecy\in\supp(F)\implies\vecx\vee\vecy\in\supp(F)$).
\item No pinning of the support of a signature in $F$ is equivalent to
  $\EQ_2$.
\end{enumerate}
\end{reptheorem}
\begin{remark*}
These conditions describe when all the
reductions in the following proof fail.  They are certainly not
exhaustive. For example the following relation $R$ is not in
Weighted-NEQ-conj, is not basically binary, and is not terraced, but
(the signature of) $R$ is IM-terraced, $R$ is closed under meets, and
has no pinning equivalent to $\EQ_2$.
\[R=\{(0,0,0),(0,0,1),(0,1,0),(1,0,0),(0,1,1)\}\]
\end{remark*}
\begin{proof}[Proof of Theorem \ref{thm:finitewts}]
By Lemma \ref{lem:freepin} we can assume $\calF$ is closed under
pinning.  We will consider each condition in turn.

\begin{enumerate}
\item Assume that $\calF$ is not IM-terraced.  The conclusion follows
  from Lemma \ref{lem:bissatwt}.

\item
  Assume that condition 2 does not hold but condition 1 holds.  Pick a
  non-terraced signature $F'\in\calF$.  By Lemma \ref{lem:minnonterr}
  there is a signature $F$ equivalent to a pinning of $F'$ and
  satisfying certain conditions: $V(F)=\{1,\cdots,|V(F)|\}$, and there
  are configurations $\vecp\in\{0,1\}^{\{1,2\}}$ and
  $\vecz\in\{0,1\}^{\{3,\cdots,|V(F)|\}}$ such that $\pin F \vecp$ is
  identically zero and for all
  $\vecx\in\{\vecp^{\{1\}},\vecp^{\{2\}}\}$ and
  $\vecy\in\{0,1\}^{\{3,\cdots,|V(F)|\}}$ we have
\[F(x_1,x_2,y_3,\cdots,y_{|V(F)|})=\begin{cases}
                   T(y_3,x_1)&\text{if $\vecy=\vecz$ or $\vecy=\overline{\vecz}$}\\
                   0&\text{otherwise}
               \end{cases}\]

  We have assumed that condition 1 holds, so $F'$ is IM-terraced, so
  $p_1=p_2$. Permuting the domain $\{0,1\}$ if necessary we can assume
  $p_1=p_2=0$ without loss of generality.

  There is a signature $G'\in\calF$ such that $\supp(G')$ is not
  closed under joins; let $G$ be a minimal pinning of $G'$ such that
  $\supp(G)$ is not closed under joins.

  By Lemma \ref{lem:minnonjoin} there exists $\vecx$ such that
  $\supp(G)=\{\veczero,\vecx,\overline{\vecx}\}$ or
  $\supp(G)=\{\vecx,\overline{\vecx}\}$.  And $\{\veczero,\vecone\}$
  is closed under joins, so $\vecx\neq\vecone$ and there is a variable
  $i$ such that $x_i=0$.  Since $G$ is IM-terraced, $\supp(G)$ is a
  delta matroid (Lemma \ref{lem:dmterr}) and hence
  $\vecx^{\{i,j\}}\in\supp(G)$ for some $j\in V(G)$. But this can only
  mean that $\vecx^{\{i,j\}}=\overline{\vecx}$, which implies
  $V(G)=\{i,j\}$. Also, the arity of $G$ is not 1. It will be harmless to
  take $V(G)=\{1,2\}$. With this assumption we have $G(1,1)=0$
  and $G(0,1),G(1,0)\neq 0$.

  Define $H\sigf {V(F)}$ by
  \[ H(x_1,x_2,\vecy)=\sum_{t=0,1} G(x_1,t)F(t,x_2,\vecy) \]
  If we shorten $F_{\{1\mapsto i,2\mapsto j\}}$ to $F_{ij}$, and
  similarly define $F'_{ij}$, and allow scalar multiplication of a
  signature by a constant, we have:
  \begin{align*}
  H_{10}&=G(1,0)F_{00}+G(1,1)G_{10}\text{ which is identically zero}\\
  H_{00}&=G(0,0)F_{00}+G(0,1)F_{10}=G(0,1)F_{10}\\
  H_{11}&=G(1,0)F_{01}+G(1,1)F_{11}=G(1,0)F_{01}
  \end{align*}
  Hence $H$ is not IM-terraced. (A related trick, expressing IMP using
  OR and NAND, is used in \cite{bdddeg}.)

  We have shown (condition 1) that there is a finite set $W$ such that
  $\Xprob\APred\nCSP^W_{=2}(\calF\cup\{H\})$; by Lemma \ref{lem:express}
  $\nCSP^W_{=2}(\calF\cup\{H\})\APred\nCSP_{=2}^W(\calF\cup\{F,G\})$;
  and by Lemma \ref{lem:freepin}
  $\nCSP^W_{=2}(\calF\cup\{F,G\})\APred\nCSP_{=2}^{W'}(\calF)$ where
  $W'=W\cup\{(0,1),(1,0)\}$.

\item
  Assume that condition 3 does not holds but conditions 1 and 2 do
  hold.  So there is a signature in $\calF$ whose support is not
  closed under joins, and a signature in $\calF$ whose support is not
  closed under meets.  By permuting the domain $\{0,1\}$ if necessary
  we can assume without loss of generality that the support of every
  signature in $\calF$ is closed under meets.

  Pick $G$ satisfying $\supp(G)=\NAND$ as follows.  Let $H$ be a
  minimal non-terraced pinning of a signature in $\calF$.  Reordering
  the variables according to Lemma \ref{lem:minnonterr}, there exist
  not necessarily distinct configurations
  $\vecy,\vecy'\in\{0,1\}^{\{3,\cdots,V(H)|\}}$ such that
  $(0,1,\vecy),(1,0,\vecy')\in\supp(H)$, and
  $(0,0,\vecy\wedge\vecy'),(1,1,\vecy\vee\vecy')$ are not both in
  $\supp(F)$.  By assumption $\supp(G)$ is closed under meets, so it must not be
  closed under joins.  By the same argument used for condition 2,
  there is a pinning $G$ of $H$ of arity 2, and we can take
  $V(G)=\{1,2\}$ so $\supp(G)=\NAND$.

  Let $h(1)=h(2)=1$ so $\supp(G_{h-\max})=\NEQ$.  Since $G_{h-\max}$
  fails condition 2, there is a finite set $W$ such that
  $\Xprob\APred\nCSP^W_{=2}(\calF\cup\{G_{h-\max}\})$.

  We will want to use variable weights that are arbitrary powers of
  two, so it is convenient to hide $W$ at this point.  By Lemma
  \ref{lem:simpleweight} there is a set of simple weightings $\calG$
  of signatures in $\calF$, and a set of simple weightings $\calG'$ of
  $G_{h-\max}$, such that
  $\nCSP^W_{=2}(\calF\cup\{G_{h-\max}\})\APred\nCSP_{=2}(\calG\cup\calG')$.
  Let \[P=\{(2^{p_0},2^{p_1})\mid p_0,p_1\in\mathbb{Z}\}\] Let
  $\calG''$ be the set of simple weightings $G'$ of $G$ satisfying
  $G'_{h-\max}\in\calG'$. In other words, for all arity 1 signatures
  $U,W$, if the signature defined by $G_{h-\max}(x,y)U(x)W(y)$ is in
    $\calG'$, then the signature defined by $G(x,y)U(x)W(y)$ is in
    $\calG''$.  Note that $|\calG''|=|\calG'|$ is finite. By Lemma
    \ref{lem:hamweight},
  \[ \nCSP_{=2}(\calG\cup\calG')\APred \nCSP^P_{=2}(\calG\cup\calG'') \]
  
  We will show that
  \begin{align}
  \nCSP^P_{=2}(\calG\cup\calG'')\APred \nCSP^{\{(1,2),(1,1),(2,1)\}}(\calG\cup\calG''\cup\{\EQ_2\})
  \label{eqn:eqtopot}
  \end{align}
  
  We are given an instance $(w,\phi)$ of $\nCSP^P_{=2}(\calG\cup\calG'')$.
  For each $v\in V=V^{\phi}$ there exists $p_v$
  such that $w(v,1)/w(v,0)=2^{p_v}$ for $i=0,1$.  Let $V'=\{v_i\mid
  v\in V; 0\leq i\leq |p_v|\}$.  Define $w':V\to W$ by $w'(v_0)=(1,1)$
  and for all $i>0$,
  \[w'(v_i)=\begin{cases}(1,2)&\text{ if $p_v<0$}\\(2,1)&\text{ if $p_v>0$}\end{cases}\]
  Modify $\phi$ as follows to obtain a new formula $\phi'$: for each
  $v\in V$, insert atomic formulas $\EQ_2(v_0,v_1)\cdots
  \EQ_2(v_{|p_v|-1},v_{|p_v|})$ and replace the two occurences of
  $v$ by $v_0$ and $v_{|p_v|}$. Note that configurations $\vecx'$ of
  $V'$ have zero weight in $(w',\phi')$ unless there exists
  $\vecx\in\{0,1\}^V$ such that $x'_{v_i}=x_v$ for all $v,i$, and in
  this case $\wt{w'}{\phi'}(\vecx')=\wt{w}{\phi}(\vecx)C$ where
  $C=\prod_{v\in V}\min(w(v,0),w(v,1))$.  Hence
  $\zwp{w}{\phi}=\zwp{w'}{\phi'}C$.  And $\zwp{w'}{\phi'}$ can be
  approximated by the oracle. This establishes the AP-reduction
  \eqref{eqn:eqtopot}.

  To finish, let $F$ be a pinning of a signature in $\calF$ such that
  $\supp(F)$ is equivalent to $\EQ_2$.  Then $F(x,y)=\EQ_2(x,y)F(x,x)$
  for all $x,y\in\{0,1\}$ so $F$ is a simple weighting of $\EQ_2$.  By
  Lemma \ref{lem:simpleweight} there is a finite set $W'$ (which we
  can assume contains $(0,1)$ and $(1,0)$) such that
  \[\nCSP^{\{(1,2),(1,1),(2,1)\}}(\calG\cup\calG''\cup\{\EQ_2\})\APred\nCSP^{W'}_{=2}(\calF\cup\{F\})\]
  and $\nCSP^{W'}_{=2}(\calF\cup\{F\})\APred\nCSP^{W'}_{=2}(\calF)$ by
  Lemma \ref{lem:freepin}.  \qedhere
\end{enumerate}
\end{proof}

%% file: highdeg.tex
\section{Degree three and higher}

In this section we will study $\nCSP^W_{\leq k}(\calF)$ for $k>2$.  We
will use a result of Sly about the complexity of the partition
function of the hardcore model on a graph.  The partition function of
the hardcore model with fugacity $\lambda$, defined on a graph $G$, is
defined to be the sum of $\lambda^{|I|}$ over independent sets $I$ of
$G$.

\begin{lemma}[\cite{sly}, Theorem 1]\label{lem:sly}
For every $d\geq 3$ there exists $\lambda_c(d),\epsilon(d)>0$ such that when
$\lambda_c(d) < \lambda < \lambda_c(d) + \epsilon(d)$, unless NP=RP,
there does not exist an FPRAS for the partition function of the
hardcore model with fugacity $\lambda$ for graphs of maximum degree at
most d.
\end{lemma}

\def\APwred{\leq_{AP}^*}

Lemma \ref{lem:sly} is not stated as an AP reduction.  We would like
to present complexity-theoretic results that are not stated as AP
reductions.  Let $\Xprob$ and $\Yprob$ be $\Rnonneg$-valued function
problems. The notation $\Xprob\APwred\Yprob$ means: $\Xprob$ has an
FPRAS if $\Yprob$ has an FPRAS. The following Lemma is given as a
remark in, for example, \cite{rcap} and \cite{jerrumbook}.

\begin{lemma}\label{lem:nprp}
If NP=RP then $\SAT$ has an FPRAS.
\end{lemma}

\begin{lemma}\label{lem:sly3}
Let $R=\NAND$ or $R=\OR$.  There exists a finite set of variable
weights $W$ such that $\SAT\APwred\nCSP_{\leq 3}^W(\{R\})$.
\end{lemma}
\begin{proof}
It suffices to consider $R=\NAND$; the definitions of $\nCSP$s are not
affected by permuting the domain $\{0,1\}$, so
\[\nCSP^{\overline W}_{\leq 3}(\overline \NAND)\APeq\nCSP^W_{\leq 3}(\NAND)\]
where $\overline W=\{(b,a)\mid(a,b)\in\calF\}$ and where $\overline{\NAND}$
is defined by $\overline{\NAND}(\vecx)=\NAND(\overline{\vecx})=\OR(\vecx)$.

Let $\lambda$ be a rational number such that
$\lambda_c(d)<\lambda<\lambda_c(d)+\epsilon(d)$ where $\lambda_c$ and
$\epsilon$ are given by Lemma \ref{lem:sly}.  Let $W=\{(1,\lambda)\}$.
We will show that $\HC_3(\lambda)\APred\nCSP_{\leq 3}^W(\NAND)$ where
$\HC_3(\lambda)$ is the problem of computing the partition function
of the hardcore model with fugacity $\lambda$ for graphs of maximum
degree at most $3$.  Then $\SAT\APwred\HC_3(\lambda)\APred\nCSP_{\leq
  3}^W(\NAND)$ by Lemma \ref{lem:sly}.

Given an instance $(V,E)$ of $\HC_3(\lambda)$, we can query the
$\nCSP_{\leq 3}^W(\NAND)$ oracle to approximate
\[ \zwp{w}{\phi}=\sum_{\vecx\in\{0,1\}^V}\left(\prod_{v\in W} \lambda^{x_v}\right)\left(\prod_{ij\in E}\NAND(x_i,x_j)\right) \]
But this is just the sum of $\lambda^{|I|}$ over independent sets $I$
in $G$, which is the correct output of $\HC_3(\lambda)$ on this
instance.
\end{proof}

\begin{reptheorem}{thm:highdeg}
Let $\calF$ be a finite set of signatures and assume that not every
signature in $\calF$ has degenerate support. There exists a finite set
of variable weights $W$ such that
$\nCSPs(\calF)\APwred\nCSP_{\leq 3}^W(\calF)$.
\end{reptheorem}

\begin{proof}
Let $F_1$ be a signature in $\calF$ whose support is non-degenerate.
Let $F_2$ be a minimal non-degenerate pinning of $F_1$.  Define
$F(x_1,x_2)=\sum_{x_3,\cdots,x_k}F_2(x_1,\cdots,x_k)$.  By Lemma
\ref{lem:minnondeg}, either the arity of $F_2$ is 2, or $\supp(F_2)$ equals
$\{\vecx,\overline{\vecx}\}$ for some tuple $\vecx$.  In either case
$R=\supp(F)$ is non-degenerate.

Now we claim that there are arity one signatures $U,V$, taking
positive values, such that for all $x,y\in\{0,1\}$ the value
$F(x,y)U(x)V(y)$ is zero or one.  Since $\supp(F)$ is not degenerate
there is a flip $F^S$ of $F$ with $F^S(0,0)=0$.  We will find
$U'(0),U'(1),V'(0),V'(1)>0$ such that $F^S(x,y)U'(x)V'(y)$ is zero-one
valued for all $x,y\in\{0,1\}$; then $F(x,y)U(x)V(y)$ is also zero-one
valued, where $U$ and $V$ are the flips $(U')^{S\cap\{1\}}$ and
$(V')^{S\cap\{1\}}$ respectively, establishing the claim.  Since $F^S$
is non-degenerate, $F^S(0,1),F^S(1,0)>0$, so
$\supp(F^S)=\{(0,1),(1,0)\}$ or $\supp(F^S)=\{(0,1),(1,0),(1,1)\}$.
In the first case take $U'(x)=1/F^S(x,1-x)$ and $V'(x)=1$ for all
$x=0,1$.  In the second case set $U'(1)=1$ and $V'(0)=1/F^S(1,0)$ and
$V'(1)=1/F^S(1,1)$ and $U'(0)=F^S(1,1)/F^S(0,1)$.

We will show that there is a finite set $W$ such that
\begin{align}\nCSPs(\calF)\APwred\nCSP^W_{\leq 3}(\calF\cup\{R\})\label{eqn:relnal}\end{align}
Then using both parts of Lemma \ref{lem:simpleweight}, there is a finite set
$W'$ such that
\[\nCSP^W_{\leq 3}(\calF\cup\{R\})\APred\nCSP^{W'}_{\leq 3}(\calF\cup\{F\})\]
But $F_1\in\calF$, and $F_2$ is a pinning of $F_1$ (Lemma
\ref{lem:freepin}), and $F$ is given by a $(\leq 3)$-formula over
$F_2$ (Lemma \ref{lem:express}):
\[\nCSP^{W'}_{\leq 3}(\calF\cup\{F\})\APred \nCSP^{W'}_{\leq 3}(\calF\cup\{F_2\})\APred \nCSP^{W'}_{\leq 3}(\calF)\]
So we are done if we can show \eqref{eqn:relnal}.

Up to equivalence,
\[ R\in\{\NAND,\OR,\EQ_2,\NEQ,\IMP\}\]
If $R=\NEQ$ then $\sum_y R(x,y)R(y,z)$ is a $(\leq 3)$-formula
expressing $\EQ_2$. By Lemma $\ref{lem:express}$ we have $\nCSP_{\leq
  3}^W(\calF\cup\{\EQ_2\})\APred \nCSP_{\leq 3}^W(\calF)$ for all sets
$W$ containing $(1,1)$. So we can ignore the case $R=\NEQ$.  If
$R=\NAND$ or $R=\OR$ then $\nCSPs(\calF)\APred\SAT\APwred
\nCSP^W_{\leq 3}(\{R\})$ for some finite set $W$, by Lemma
\ref{lem:sateasy} and Lemma \ref{lem:sly3}.  Otherwise $R=\EQ_2$ or
$R=\IMP$. We will ``3-simulate equality'' as in \cite{bdddeg}.

By Lemma \ref{lem:tract} we may assume that $\calF$ is not contained
in Weighted-NEQ-conj.  It follows from \cite[Theorem 14, Proposition
  25]{lsm} that there is a finite set $S$ of arity 1 signatures such
that $\nCSP(\calF\cup S)\APeq \nCSP(\calF\cup\{\IMP\})\APeq
\nCSPs(\calF)$. Let \[W=\{(U(0),U(1))\mid \text{$U$ has arity 1, and
  $U\in \calF$}\}\cup\{(1,1)\}\]

We will show that $\nCSP(\calF\cup S)\APred \nCSP^W_{\leq 3}(\calF)$.
Given an instance $(V,\phi)$ of $\nCSP(\calF\cup S)$, for each
variable $v$ replace all its occurences by distinct variables
$v_1,\cdots,v_d$ and insert new atomic formulas
$R(v_1,v_2)\cdots R(v_d,v_1)$.  This gives a new formula
$\phi'$ on variables $V'$. Now replace any arity 1 atomic formula $U(v_i)$
by a variable weight on $v_i$; that is, delete these atomic formulas
to obtain $\phi''$
and define $w:V'\to V$ by
\[ w(v_i)=\begin{cases}(U(0),U(1))&\text{ if there is an atomic formula $U(v_i)$ in $\phi'$}\\
                     (1,1)&\text{ otherwise}\end{cases}\] for all
$v_i\in V'$. Then $\phi''$ is a ($\leq 3$)-formula with
$\zwp{w}{\phi'}=\zp{\phi}$, so we can just query the oracle.
\end{proof}

%% file: fpras.tex
\section{Tractable problems not in FP}
In this section we will argue that there is a large tractable region
for $\nCSPs_{\leq d}$.  The existence of these FPRASes contrasts with
the unbounded problem $\nCSPs$. Assuming that $\BIS$ does not have an
FPRAS, $\nCSPs(F)$ has an FPRAS if and only if $\nCSPs(F)$ is in FP, as
least as long as $F$ is rational-valued (see Lemma \ref{lem:unbounded} and
\cite{csptrich}). But $\nCSPs_{\leq d}(F)$ can
have an FPRAS even when $\nCSPs_{\leq d}(F)$ is $\#P$-hard.

\begin{proposition}\cite[Theorem 5.3]{hpacc}
If $\mathcal F$ is not a subset of Weighted-NEQ-conj then
$\nCSPs_{\leq 3}(\calF)$ (without any approximation) is $\#P$-hard.
\end{proposition}
\begin{proof}
Define $U$ by $U(0)=1$ and $U(1)=2$.  Using variable weights instead
of $U$, we have $\nCSP_{\leq 3}(\calF\cup\{U\})\APred\nCSPs_{\leq
  3}(\calF)$.  Now we can appeal to \cite[Theorem 5.3]{hpacc}.  Their
set ``$\mathcal A$'' does not contain $U$, and Weighted-NEQ-conj is
contained in their ``$\mathcal P$''.  Hence $\nCSP_{\leq 3}(\{F, U\})$
is $\#P$-hard.
\end{proof}

This following argument is inspired by \cite{jztwostate}, and in
particular we use the same quantity $J$.

\begin{reptheorem}{thm:boxg}
Let $d,k\geq 2$.  Let $F$ be a an arity $k$ signature with
values in the range $[1,\frac{d(k-1)+1}{d(k-1)-1})$.  Then
  $\nCSPs_{\leq d}(F)$ has an FPRAS.
\end{reptheorem}
\begin{proof}
We will use a path coupling argument on a Markov chain with Glauber
dynamics.  We will proceed by giving a FPAUS, which in this case is a
randomised algorithm that, given an instance $(w,\phi)$ and
$\epsilon>0$, outputs a random configuration $\mu$ such that the total
variation distance of $\mu$ from $\pi^w_\phi$ is at most $\epsilon$
where $\pi^w_\phi(\sigma)=\wt{w}{\phi}(\sigma)/Z^w_\phi$; and the
algorithm runs in time polynomial in the size of the input and
$\log(1/\epsilon)$.

The FPAUS is to simulate a Markov chain of configurations
$(X_t)_{t=0,1,\cdots}$ and output $X_T$ for some $T$ to be determined
later.  For configurations $X$ and variables $v$ we will use the
notation $X[v\mapsto j]$ to mean $X[v\mapsto j](u)=X(u)$ for $u\neq v$
and $X[v\mapsto j](v)=j$.  Let $X_0\in \{0,1\}^V$ be any
configuration. For each $t\geq 1$ let $v_t$ be distributed uniformly
at random and let $X_t$ be distributed according to heat bath
dynamics, that is, distributed according to $\pi^w_\phi$
conditioned on $X_t\in \{X_{t-1}[v_t\mapsto 0],X_{t-1}[v_t\mapsto
  1]\}$. Thus
\[
  \mathbb{P}[X_t(i)=1 \mid X_{t-1},v_t] = \frac{\wt{w}{\phi}(X_{t-1}[v\mapsto
      1])}{\wt{w}{\phi}(X_{t-1}[v\mapsto 0])+\wt{w}{\phi}(X_{t-1}[v\mapsto 1])}
\]
This probability is easy to compute exactly, so each step of the
Markov chain can be simulated efficiently.

Consider another Markov chain $(Y_t)_{t\geq 0}$ distributed in the
same way as $(X_t)_{t\geq 0}$, with the optimal coupling given that
both chains choose the same variables $v_t$. So
\[
\mathbb{P}[X_t(v_t)\neq Y_t(v_t)|X_{t-1},Y_{t-1},v_t]=|\mathbb{P}[X_t(v_t)=1|X_{t-1},Y_{t-1},v_t]-\mathbb{P}[Y_t(v_t)=1|X_{t-1},Y_{t-1},v_t]|
\]
Define $\beta=\beta(w,\phi)=\max_{X_0,Y_0:|X_0\setdiff Y_0|=1}
\mathbb{E}[d(X_1,Y_1)]$.
Let $M$ be the maximum value taken by $F$.
We will establish the bound
\begin{align}\label{eqn:betabound}\beta\leq 1-c|V|^{-1}\end{align}
for some $c>0$ depending only on the parameters $d,k,M$.  Then by the
General Path Coupling Theorem of \cite{pathcouple} the total variation
distance from the stationary distribution is at most $\epsilon$ as
long as $T\geq \log(|V|\epsilon^{-1})/\log
\beta^{-1}=\operatorname{poly}(|V|,\log \epsilon^{-1})$. This gives
the required FPAUS.  Given the FPAUS, there is an FPRAS by
\cite[Theorem 6.4]{jvv} (the self-reducibility is Lemma \ref{lem:freepin}).

We will now bound $\beta$.  Fix configurations $X_0$ and $Y_0$ that
only differ on a single variable $u$. For all $v_1\in V$ define

\[ E(X_0,Y_0,v_1)=|\mathbb{P}[X_1(v_1)=1|X_0,Y_0,v_1]-\mathbb{P}[Y_1(v_1)=1|X_0,Y_0,v_1]| \]

Define $W_{ij}=W_{ij}(X_0,Y_0,v_1)=\wt{w}{\phi}(X_0[u\mapsto i][v_1\mapsto j])$ for all $i,j\in\{0,1\}$. Then
\begin{align*}
E(X_0,Y_0,v_1)&=|\mathbb{P}[X_1(v_1)=1|X_0,Y_0,v_1]-\mathbb{P}[Y_1(v_1)=1|X_0,Y_0,v_1]|\\
&=\left|\frac{W_{01}}{W_{00}+W_{01}}-\frac{W_{11}}{W_{10}+W_{11}}\right|\\
&=\frac{|W_{00}W_{11}-W_{01}W_{10}|}{W_{00}W_{11}+W_{01}W_{10}+W_{00}W_{10}+W_{01}W_{11}}\\
&\leq\frac{|W_{00}W_{11}-W_{01}W_{10}|}{W_{00}W_{11}+W_{01}W_{10}+2\sqrt{W_{00}W_{10}W_{01}W_{11}}}\\
&=\frac{|\sqrt{W_{00}W_{11}}-\sqrt{W_{01}W_{10}}|}{\sqrt{W_{00}W_{11}}+\sqrt{W_{01}W_{10}}}
\end{align*}

Let $v_1\in V\setminus\{u\}$.  Denote by $I'(u,v_1)\subseteq I$ the
set of indices of atomic formulas with $u$ and $v_1$ in their scope.
For all $i\in I'(u,v_1)$ and all $j,k\in\{0,1\}$, define
\[ F'_i(j,k)=F(x_{\scope(i,1)},\cdots,x_{\scope(i,k)})\]
where $x_v=(X_0[u\mapsto j][v_1\mapsto k])_v$. Define
$W'(j,k)=\prod_{i\in I'} F'_i(j,k)$.  The other weights depend on $u$
or $v_1$ alone, so $W'(0,0)W'(1,1)/W'(0,1)W'(1,0)$ equals $W_{00}W_{11}/W_{01}W_{10}$ and
\[E(X_0,Y_0,v_1)\leq\frac{|\sqrt{W'(0,0)W'(1,1)}-\sqrt{W'(0,1)W'(1,0)}|}{\sqrt{W'(0,0)W'(1,1)}+\sqrt{W'(0,1)W'(1,0)}} \]

For all arity 2 signatures $G$ (taking strictly positive values),
define $J(G)=\frac{1}{4}\log\frac{G(0,0)G(1,1)}{G(0,1)G(1,0)}$.  Note
that the functions $F'_i$ take values in the range $[1,M]$ so
$|J(F'_i)|\leq \frac 1 2 \log M$; also recall that $\tanh$ is
non-decreasing and subadditive for positive reals, that is,
$\tanh(x+y)=\frac{\tanh x+\tanh y}{1+\tanh x \tanh y}\leq
\tanh(x)+\tanh(y)$. Hence

\begin{align*}
E(X_0,Y_0,v_1)&\leq\frac{|\sqrt{W'(0,0)W'(1,1)}-\sqrt{W'(0,1)W'(1,0)}|}{\sqrt{W'(0,0)W'(1,1)}+\sqrt{W'(0,1)W'(1,0)}}\\
     &= \tanh |J(W')|\\
     &= \tanh\left|\sum_{i\in I'(u,v_1)} J(F'_i)\right|\\
     &\leq |I'(u,v_1)|\tanh\left(\frac{1}{2}\log M\right)\\
     &= |I'(u,v_1)|\frac{M-1}{M+1}
\end{align*}

The variable $u$ appears in at most $d$ atomic formulas, each of which
contributes at most $k-1$ to $\sum_{v_1}|I'(u,v_1)|$. Rearranging
$M<\frac{d(k-1)+1}{d(k-1)-1}$ we get
$d(k-1)\frac{M-1}{M+1} < 1$, so
\begin{align*}
\mathbb{E}[d(X_1,Y_1)]=1-\frac{1}{|V|}+\frac{1}{|V|}\sum_{v_1\in V\setminus\{u\}}E(X_0,Y_0,v_1)\\
                      \leq 1-\left(1-d(k-1)\frac{M-1}{M+1}\right)/|V|
\end{align*}
giving the required bound \eqref{eqn:betabound}.
\end{proof}

\section{Infinite sets of variable weights are sometimes necessary}\label{sec:monodim}
Theorem \ref{thm:finitewts} gives some circumstances in which the set
of variable weights in Theorem \ref{thm:weightthm} can be taken to be
finite.  On the other hand, assuming that $\PMP$ does not have an
FPRAS, there is a situation where we cannot take the set of variable
weights to be finite.

\def\zmd{Z_{\mathrm{MD}}}

Let $G$ be a (simple) graph with a non-negative edge weight
$\lambda(e)$ for each edge $e$ of $G$. Recall that a matching in $G$
is a subset $M$ of the edge set of $G$ such that no two edges in $M$
share a vertex.  The partition function $\zmd(G)$ of the
monomer-dimer model on $G$ is the sum, over all matchings $M$ in $G$,
of $\prod_{e\in M}\lambda(e)$.

\begin{lemma}[\cite{approxperm}, Corollary 3.7]\label{lem:monomerdimer}
There is an FPRAS for the partition function of the monomer-dimer
model if the edge weights are given as integers in unary.
\end{lemma}

Let $R$ the the relation $\{(0,0,0),(0,0,1),(0,1,0),(1,0,0)\}$.

\begin{proposition}
Let $W$ be a finite set of integer-valued variable weights. Then
$\nCSP^W_{=2}(\{R\})$ has an FPRAS.
\end{proposition}

\begin{proof}
We can scale the variable weights to assume that $w(0)\in\{0,1\}$ for
all $w\in W$.  We will give an AP-reduction from $\nCSP^W_{=2}(\{R\})$ to
the problem of computing the monomer-dimer partition function of a
graph with positive edge weights specified in unary.  Let $(w,\phi)$
be an instance of $\nCSP^W_{=2}(\{R\})$.

Let $G$ be the edge-weighted multigraph whose vertices are atomic
formula indices $I^\phi$ and with, for each $v\in V$ with $w(v,0)=1$,
an edge with weight $w(v,1)$ joining the two indices of the atomic
formulas in which $v$ appears - and if a variable is used twice in the
same atomic formula then we get a vertex with a loop.  For each
$v\in V$ with $w(v,0)=0$ delete the two vertices corresponding to the
atomic formulas in which $v$ appears.

The definition of the partition function for the monomer-dimer model
extends to multigraphs, and the value of the instance $(w,\phi)$ is
$\zmd(G)$: positive-weight configurations $\sigma:V\to \{0,1\}$ of
$\zwp{}{\phi}$ correspond to subsets $M=\sigma^{-1}(1)$ of the edge set
of $G$ that are matchings, and the weight $\wt{w}{\phi}(\vecx)$ is the
weight $\prod_{e\in M}\lambda(e)$ of the corresponding matching $M$.
We can transform this multigraph to a simple graph without changing
the partition function: a set of parallel edges with weights
$w_1,\cdots,w_k$ are equivalent to having a single edge with weight
$w_1+\cdots+w_k$, and any loop can be deleted.

The result of these transformations is a simple edge-weighted graph
$G'$ such that $\zmd(G)=\zwp{w}{\phi}$, and whose edge weights
belong to $\sum_{w\in W'}w(1)$ for some subset $W'$ of $W$.  But there
are finitely many such edge weights, so we can write any such edge
weight in unary in constant time, and use the oracle to approximate
$\zmd(G')$.
\end{proof}

By Theorem \ref{thm:basic} however, $\PMP\APred\nCSPs_{=2}(\{R\})$.